\documentclass[11pt]{amsart}
\usepackage{latexsym, amsmath, amssymb,amsthm,amsopn,amsfonts,mathrsfs,tikz}
\usepackage{stmaryrd}
\usepackage{version}
\usepackage{epsfig,graphics,color,graphicx,graphpap,dsfont,eufrak}
\usepackage{amssymb}

\usepackage{tikz}

\usetikzlibrary{snakes}

\ifx\pdfoutput\undefined
  \DeclareGraphicsExtensions{.eps}
\else
  \ifx\pdfoutput\relax
    \DeclareGraphicsExtensions{.eps}
  \else
    \ifnum\pdfoutput>0
      \DeclareGraphicsExtensions{.pdf}
    \else
      \DeclareGraphicsExtensions{.eps}
    \fi
  \fi
\fi

\allowdisplaybreaks

\allowdisplaybreaks

\newcommand{\LV}{\left|}
\newcommand{\RV}{\right|}

\newcommand{\LB}{\left[}
\newcommand{\RB}{\right]}
\newcommand{\LC}{\left(}
\newcommand{\RC}{\right)}

\newcommand{\LCB}{\left\{}
\newcommand{\RCB}{\right\}}
\newcommand{\p}{\partial}
\newcommand{\R}{\mathbb{R}} 
\newcommand{\vertiii}[1]{{\left\vert\kern-0.25ex\left\vert\kern-0.25ex\left\vert #1 
    \right\vert\kern-0.25ex\right\vert\kern-0.25ex\right\vert}}

\newcommand{\be}{\begin{equation}}
\newcommand{\ee}{\end{equation}}

\newcommand{\bse}{\begin{subequations}}
\newcommand{\ese}{\end{subequations}}

\newcommand{\jump}[1]{\left\llbracket{#1}\right\rrbracket}

\DeclareFontFamily{OT1}{pzc}{}
\DeclareFontShape{OT1}{pzc}{m}{it}{<-> s * [1.10] pzcmi7t}{}
\DeclareMathAlphabet{\mathpzc}{OT1}{pzc}{m}{it}

\theoremstyle{plain}

\newtheorem{thm}{Theorem}[section]
\newtheorem{prop}{Proposition}[section]
\newtheorem{cor}[prop]{Corollary}

\theoremstyle{definition}

\newtheorem{assumption}{Assumption}

\newtheorem{rem}{Remark}[section]

\numberwithin{equation}{section}

\def\squarebox#1{\hbox to #1{\hfill\vbox to #1{\vfill}}}

\usepackage{amsxtra}

\def\ba{\begin{array}}
\def\ea{\end{array}}
\def\bea{\begin{eqnarray}}
\def\eea{\end{eqnarray}}
\def\beas{\begin{eqnarray*}}
\def\eeas{\end{eqnarray*}}
\def\bi{\begin{itemize}}
\def\ei{\end{itemize}}

\def\p{\partial}

\def\({\textnormal{(}}
\def\){\textnormal{)}}
\def\[{[\neg[}
\def\]{]\neg]}

\def\neg{\negthinspace}

\def\b1{{\bf 1}}

\usepackage{hyperref}

\title[Surface recovery from pressure]{Reconstruction of stratified steady water waves from pressure readings on the ocean bed}

\author[R.~M. Chen]
{Robin Ming Chen}
\address{Robin Ming Chen\newline
Department of Mathematics\\
University of Pittsburgh\\
Pittsburgh, PA 15260} \email{mingchen@pitt.edu}
\author[S. Walsh]
{Samuel Walsh}
\address{Samuel Walsh \newline
Department of Mathematics \\ 
University of Missouri\\
Columbia, MO 65211}
\email{walshsa@missouri.edu}

\begin{document}

\begin{abstract}
Consider a two-dimensional stratified solitary wave propagating through a body of water that is bounded below by an impermeable ocean bed.  In this work, we study how such a wave can be reconstructed from data consisting of the wave speed, upstream and downstream density profile, and the trace of the pressure on the bed.  First, we prove that this data uniquely determines the wave, both in the (real) analytic and Sobolev regimes.  Second, for waves that consist of multiple layers of constant density immiscible fluids, we provide an exact formula describing each of the interfaces in terms of the data.    Finally, for continuously stratified fluids, we detail a  reconstruction scheme based on approximation by layer-wise constant density flows.  
\end{abstract}
\maketitle

\section{Introduction} \label{intro section}

One of the primary means of measuring waves at sea are \emph{pressure transducers} --- recording devices distributed throughout the ocean that read the subsurface pressure.  There are many advantages to this approach:  pressure transducers consume relatively little energy, are inexpensive to build and maintain, and provide a wealth of {\it in situ} data.  As they are deployed submerged, moreover, they have a minimal impact on marine life and ocean traffic. 

The use of pressure transducers raises several intriguing mathematical questions:  How much information can really be gleaned about the flow from just the trace of the pressure on the ocean bed?  Is it possible to totally reconstruct a passing traveling wave from this data?  These issues are the subject of the present paper.  Specifically, we are concerned with the role density stratification may play and whether pressure transducers can be used to track internal waves.  

Let us first discuss the situation for constant density fluids. If one takes the linearized free surface Euler equations for a small-amplitude gravity wave,  it is simple to show that the  surface profile can be expressed in terms of the pressure trace; indeed, this mapping is a Fourier multiplier whose symbol is referred to as the (linear) \emph{pressure transfer function} (see, e.g., \cite{bishop1987measuring,escher2008recovery,kinsman1965wind,kuo1994transfer}).  More concretely, for a sinusoidal wave with wavenumber $k$, the connection between the the dynamic pressure $\mathfrak{p}$  and the surface displacement $\eta$ can be written as 
\be\label{linear}
\mathfrak{p} = \rho g {\cosh(kh_m) \over \cosh(kd)} \eta,
\ee
where $h_m$ is the height of the pressure meter above the bottom, $d$ is the mean depth, and the dynamic pressure $\mathfrak{p}$ is defined to be the difference between the pressure and the hydrostatic pressure, that is,
$$\mathfrak{p} = P - P_{\textrm{atm}} - \rho g (d - h_m). $$
For pressure data obtained from the sea-bed, $h_m=0$ and hence \eqref{linear} becomes
$$\mathfrak{p} = {\rho g \over \cosh(kd)} \eta,$$
which, in the zero depth limit $d\to 0$, recovers the hydrostatic approximation
\be\label{hydrodynamic}
\mathfrak{p} = \rho g \eta.
\ee

However, it is well known that nonlinear effects play an important role in the surface reconstruction problem for shallow water waves or for waves in the surf zone (see \cite{bergan1968wave,bishop1987measuring,tsai2005recovery}, for instance).  Experiments suggest that the predictions of  linear models such as \eqref{linear} may  diverge sharply from observation \cite{cavaleri1980wave}. For waves of large or even moderate amplitude, the discrepancy can be significant \cite{bishop1987measuring,tsai2005recovery}. Another issue with the linear theory is that, due to its reliance on Fourier series techniques, it can only treat periodic traveling waves.     

These limitations have motivated a recent push to obtain reconstruction formulas that account for nonlinear effects and which apply to solitary waves.  Nonlinear nonlocal equations relating the trace of the dynamic pressure on the bed to the surface profile of a solitary wave were established by several groups of researchers (see \cite{constantin2012pressure,deconinck2011recovering,oliveras2012recovering}). Notably, these papers work directly from the free surface Euler equations without further approximation.  In \cite{clamond2013new,constantin2013pressure}, the authors obtained implicit but exact and tractable relations for periodic traveling waves and performed straightforward numerical procedures to derive the free surface from the pressure at the bed. In \cite{constantin2014estimating}, Constantin used a strong maximum principle argument completely different from the aforementioned approach to derive some improved estimates for the wave height of periodic traveling waves.

The above body of work focuses on irrotational steady waves propagating through constant density water.  Actual waves in the  ocean, however, frequently exhibit a heterogeneous density due to salinity or temperature gradients.  This {density stratification} generates vorticity within the flow as the effects of gravity and inertia are experienced differently by heavier and lighter fluid particles.  Moreover, a heterogeneous density distribution can permit large-amplitude waves to propagate in the bulk even while the surface remains relatively undisturbed.  These so-called \emph{internal waves} have been observed in oceanic areas for many years (cf., e.g. \cite{perry1965large}, or the compendium in \cite{jackson2004atlas}). They are known to have important implications for particle transport \cite{kingsford1986influence,shanks1983surface,shanks1987onshore}, mixing and energy dissipation \cite{inall2000impact,kumar2010internal,mackinnon2003mixing,moum2003structure}, as well as affecting acoustic propagation \cite{preisig1997coupled,rubenstein1999observations,zhou1991resonant}. 

The vorticity that naturally accompanies stratification poses a serious mathematical challenge:  nearly all of the irrotational theory is built on tools (e.g., conformal mappings, nonlocal reformulations via Dirichlet--Neumann operators, variational principles, etc.) that do not have obvious analogues in the rotational regime.  As a consequence, rigorous results on rotational waves have begun to appear in the literature only quite recently. Assuming the vorticity is constant throughout the  fluid domain, which is the simplest model of a wave-current interaction, Ali and Kalisch \cite{ali2013reconstruction} derived a relationship in the long-wave limit that gives a direct map between the surface elevation and the pressure at the bottom of the fluid and within the bulk of the fluid domain. Later Visan and Oliveras \cite{vasan2014pressure} obtained an exact reconstruction formula in this regime without approximation. 

At present, no one to our knowledge has discovered how to adapt these ideas to the case of an \emph{arbitrary} vorticity distribution.  Indeed, constant vorticity is a very special (though physically significant) case  that rarely provides much insight into how to attack the general problem.   However, in \cite{henry2013pressure}, D. Henry  proved that, if the vorticity and pressure are real analytic, then the trace of the pressure on the bed uniquely determines a solitary rotational wave. The analyticity condition is needed because he employs a Cauchy-Kowalevski argument. In that sense, the surface reconstruction problem is well-defined even for rotational waves so long as they are extremely regular.

{\bf Unique determinability with stratification.} Our first results pertain to the question of unique determinability of a stratified steady wave from the trace of the pressure on the bed.  Mathematically, stratification plays a somewhat similar role in the governing equation for steady waves as vorticity.  This suggests that Henry's method can be adapted to analytically stratified flows.  Following this strategy, we prove that, if two stratified traveling waves share the same wave speed, analytic streamline density function, and analytic pressure trace on the bed, then they must coincide exactly; see Theorem \ref{thm_determinability} and Corollary \ref{cor_periodic}.  

However, analyticity is a very strong assumption that is particularly ill-suited to stratified waves.  Water columns in the ocean are frequently observed to have nearly constant density outside of thin transition layers called pycnoclines.  Thus, as it passes through a pycnocline, the density experiences something close to a jump discontinuity.  Indeed, a common practice is to view these waves as being a layering of multiple immiscible fluids each with its own constant density. In this model, the density distribution is piece-wise analytic, allowing us to iterate the Cauchy-Kowalevski approach of Henry and determine all of the free surfaces (cf. Corollary \ref{cor_layered}).  

The more interesting question is whether one has unique determinability for a continuous but not necessarily analytic density (e.g., the existence theory in \cite{walsh2009stratified} is formulated for density of class $C^{1, \alpha}$).  For this case, Henry's argument does not seem to apply and so an entirely new idea is needed.   In Theorem \ref{thm_determinability general}, we confirm that unique determinability still holds in the Sobolev space setting.  Our method is based on reformulating the problem as an elliptic equation on a strip which can be treated using the theory of strong unique continuation (cf. Theorem \ref{thm_strong continuation}).  In fact, this approach can be used for constant density rotational waves as well, hence Theorem \ref{thm_determinability general} generalizes Henry's work to the Sobolev regularity regime (see Remark \ref{henry remark}).

{\bf Surface reconstruction.} The results above show that surface reconstruction from pressure problem is mathematically well-defined for stratified waves, even those with a more realistic degree of smoothness.  As the second contribution of this paper, we present an explicit scheme for performing this reconstruction.  In order to explain our ideas, first we must discuss the physical setup in more detail.

We are interested in two-dimensional solitary waves with heterogeneous density $\varrho$, limiting to uniform flows upstream and downstream, and traveling with a constant speed $c$. Shifting to a moving reference frame eliminates time dependence from the system and thus we may assume the wave occupies a steady fluid domain
\be \label{def Omega} \Omega = \{ (x,y) \in \mathbb{R}^2 : -d < y < \eta(x) \}, \ee
where the unknown function $\eta$ is the free surface profile.  We call the fluid \emph{continuously stratified} provided that $\varrho \in C(\overline{\Omega})$. We also assume that $\varrho > 0$ in $\overline{\Omega}$ and that the fluid is stably stratified, i.e.,  $y \mapsto \varrho(\cdot, y)$ is non-increasing.

Continuously stratified fluids are difficult to analyze directly since they are rotational.  We therefore begin by studying the \emph{layered model} wherein  $\Omega$ is partitioned into finitely many immiscible fluid regions
\be \label{def Omega_i} {\Omega} = \bigcup_{i=1}^N {\Omega_i}, \quad \Omega_i := \{ (x, y) \in \Omega : \eta_{i-1}(x) < y < \eta_i(x) \} \ee
with the density $\varrho_i := \varrho|_{\Omega_i}$ constant in each layer.  Here we index so that $\eta_0 := -d$, $\eta_N := \eta$ and $\Omega_i$ lies beneath $\Omega_{i+1}$, for $i = 1, \ldots N-1$; see Figure \ref{layered figure}.   Note that this is implicitly assuming that the free surfaces all have graph geometry.  

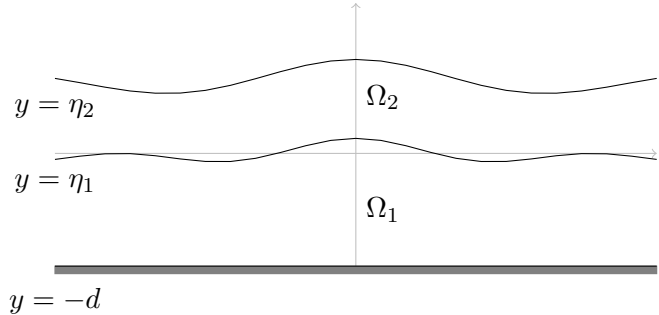
\begin{figure} \label{layered figure}
\begin{center}
\begin{tikzpicture}
\draw [thin,lightgray,->] (-4,0) -- (4,0); 							
\draw [thin,lightgray,->] (0,-1.5) -- (0,2) ;							

\draw [fill=gray,thin,gray] (-4,-1.6) rectangle (4,-1.5);				
\draw [thin,-] (-4, -1.5) -- (4, -1.5);

\draw [domain=-4:4] plot (\x, {0.10*cos(\x r)+0.10*cos(2*\x r)});		
\draw [domain=-4:4] plot (\x, {1.0+0.15*cos(\x r)+0.10*cos(1.5*\x r)});		
\node [ right] at (0,-0.75) {$\Omega_1$};
\node [ right] at (0,0.75) {$\Omega_2$};

\node [below  ] at (-4,-0.15) {$y = \eta_1$};
\node [below  ] at (-4,0.85) {$y = \eta_2$};
\node [below] at (-4,-1.65) {$y= -d$};
\end{tikzpicture}
\end{center}
\caption{Example with two fluid layers, $\Omega_1$ and $\Omega_2$.}
\end{figure}

The first task is to develop a procedure that allows us to recover a layer-wise constant density irrotational solitary waves from its pressure trace on the bed.  For this, we can exploit the fact that the velocity potential $\varphi$ and standard stream function $\psi_s$ are harmonic conjugates within each layer $\Omega_i$.  Since the density is constant,  $\varphi$ and the \emph{pseudo-stream function} $\psi := \sqrt{\varrho_i} \psi_s$ are harmonic conjugates as well.  The pseudo-stream function is better adapted to the study of stratified flows than the standard stream function.  Then, the conformal change of variables 
$$(x, y) \mapsto \LC -\frac{1}{c}\varphi(x,y), \, -\frac{1}{c} \psi(x,y) \RC, $$
transforms each layer $\Omega_i$ to a fixed horizontal strip (cf. Section \ref{subsec_conformal}). At this point, an adaptation of Constantin's method \cite{constantin2012pressure} allows us to reconstruct the flow in the first layer up to the first interface.   However, when we attempt to continue the process and compute the flow in the second layer,  we are confronted with an unfortunate feature of the conformal coordinates: because the velocity potential is in general discontinuous over the interface, the conformal variables behave badly as one crosses from one fluid region to the next.  This necessitates some additional work, but can be overcome with a further reparmeterization of the data on the interface, and by paying careful attention to the jump conditions there (cf. Section \ref{subsec repara}).  Ultimately, we are able to obtain an explicit reconstruction that determines the flow in every layer, as well as all of the free surfaces, directly from the trace of the pressure on the bed.

Finally, we return to the case of continuously stratified waves using an approximation argument.  As discussed above, the layered model is a convenient idealization of a continuously stratified fluid.  In a very recent work, the authors showed that it can in fact be made \emph{rigorous} for periodic traveling waves in a certain small-amplitude regime (see \cite{ChenWalsh2014continuity}).  That is, one can show that as  $\sum^N_{i=1} \varrho_i \mathds{1}_{\Omega_i}$ converges to $\varrho$ in $L^\infty$, the corresponding layer-wise smooth waves converge to the continuously stratified one.  On the other hand, solitary waves in this same regime can be realized as the limit of periodic traveling waves taking the period to infinity  \cite{turner1984variational}.  Using a diagonalization argument, we are therefore able to reconstruct a continuously stratified solitary wave as a limit of  layer-wise constant waves considered in the previous paragraph; see the discussion in Section \ref{subsec layer approx}.

\section{Formulation}\label{formulation section}

In this section, we present several formulations of the steady stratified water wave system, each one tailored to one of the problems we later consider.  Their equivalence is fairly straightforward to show (cf., e.g.,  \cite[Lemma A.2]{ChenWalsh2014continuity} where it is proved for much weaker regularity solutions).

\subsection{Eulerian formulation}\label{subsec eulerian}  Consider a two-dimensional wave in a stratified water propagating to the right with speed $c > 0$.  We assume that the only restorative force is gravity.  Adopting a moving reference frame, we can view the system as time-independent.  The fluid domain $\Omega$ is a priori unknown, but, as in \eqref{def Omega}, we assume that it lies below the graph of some smooth function $\eta$ and above an impermeable ocean bed at $\{ y = -d \}$.  

Let us now formulate the governing equations for a multi-layered wave, which includes the family of continuously stratified fluid as a special case.  With that in mind, suppose that $\Omega$ is organized into immiscible strata as in \eqref{def Omega_i}.  Then the wave is described mathematically by a velocity field $(u,v) : \Omega \to \mathbb{R}^2$, a density function $\varrho : \Omega \to \mathbb{R}_+$, and a pressure $P : \Omega \to \mathbb{R}$.  We follow the convention that, for a quantity $f$ with domain $\Omega$, $f_i := f|_{\Omega_i}$.  

We are interested in $(u,v,\varrho)$ for which each restriction $(u_i, v_i, \varrho_i)$ is reasonably smooth in each $\overline{\Omega_i}$, but we do not impose any continuity assumptions on $\overline{\Omega}$.  By contrast, $P$ will be continuous throughout the fluid domain.   Note that at this stage we do not require that $\varrho_i$ is a constant.  

To represent a traveling water wave with wave speed $c$, $(u,v,\varrho, P)$ must solve the incompressible Euler system:  conservation of mass
\bse  \label{weakeuler}\be (u-c) \varrho_x + v \varrho_y  =  0, \qquad \textrm{in } \Omega_i, \label{weakmass}  \ee
conservation of momentum
\begin{align}
\varrho (u-c) u_x + \varrho v u_y  & =  -P_x , \qquad \textrm{in } \Omega_i, \label{weakmomentumx} \\
\varrho (u-c) v_x + \varrho v v_y & =  -P_y - g \varrho, \qquad \textrm{in } \Omega_i, \label{weakmomentumy} \end{align} 
and incompressibility
\be u_x + v_y =  0, \qquad \textrm{in } \Omega_i. \label{weakvolume}  \ee
\ese
Here $g$ the gravitational constant of acceleration. We say a fluid is stably stratified if 
\be y \mapsto \varrho(\cdot,y) \textrm{ is non-increasing.} \label{stablestratificationeuler} \ee

The motion of the interfaces is driven by the kinematic boundary condition:
\bse \label{weakeulerboundary} \be 
v =  (u-c) \partial_x \eta_i, \qquad \textrm{on } \{y = \eta_i(x)\}, \label{weakkinematicinterface} \ee
Note that this includes the air-sea interface.  Also implicit above is the fact that 
$$ \frac{v_i}{u_i-c} = \frac{v_{i+1}}{u_{i+1}-c} \qquad \textrm{on } \{ y = \eta_i(x)\},$$
which follows from the continuity of the normal component of the velocity over the interface.  Since the ocean bed is impermeable, we must have that 
\be v =  0, \qquad \textrm{on } \{y = -d\}. \label{weakkinematicbed} \ee

Finally, for the pressure we require  
\be P =  \displaystyle P_{\textrm{atm}}, \qquad \textrm{on } \{y = \eta(x)\}, \label{weakdynamic} \ee
where $P_{\textrm{atm}}$ is the (constant) atmospheric pressure.  This choice effectively neglects surface tension.

It will be important for our later reformulations to assume that there is no horizontal stagnation in the flow:
\be u-c < 0 \qquad \textrm{in } \overline{\Omega} \label{nostagnation}. \ee \ese
Looking at \eqref{weakmomentumx}--\eqref{weakmomentumy}, it is clear that \eqref{nostagnation}  prohibits a certain degeneracy in the system.  

In this work, we mostly consider solitary waves.  This means that we impose the asymptotic conditions 
\be (u,v) \to (\mathring{u},0),~ \eta(x) \to 0 \qquad \textrm{as } |x| \to \infty. \label{upstream condition} \ee
Here $\mathring{u} = \mathring{u}(y)$ is a given function.  In most cases, we will assume that $\mathring{u} = 0$, i.e., there is uniform velocity at infinity.  For the time being, though, we state things in their full generality.  

\subsection{Stream function formulation}\label{stream function subsection}
The Euler equations are a complicated nonlinear system.  It is often far more convenient to reformulate them in terms of a scalar quantity.  Observe that the conservation of mass \eqref{weakmass} and incompressibility \eqref{weakvolume} ensure that we may define a function $\psi = \psi(x,y)$ by  
\be \psi_x = -\sqrt{\varrho}v,\qquad \psi_y = \sqrt{\varrho} (u-c)  \qquad \textrm{in } \Omega. 
\label{defpsi} \ee
This is the pseudo (relative) stream function for the flow, but throughout this paper we will simply refer to it as the \emph{stream function}. From \eqref{defpsi} and \eqref{nostagnation}, we see that the no stagnation condition translates to
\be \psi_y < 0 \qquad \textrm{in } \Omega. 
\label{nostagnationpsi} \ee

The level sets of $\psi$ are called the \emph{streamlines}. In fact \eqref{weakkinematicinterface}--\eqref{weakkinematicbed} directly imply that the free surface, internal interfaces, and ocean bed are all streamlines. Since \eqref{defpsi} only determines $\psi$ up to a constant in each $\Omega_i$, we may take $\psi$ to be continuous in $\overline{\Omega}$, and set  $\psi = 0$ on the air--sea interface.  Then $\psi = -p_0$ on the bed $\{y = -d\}$, where $p_0$ is the \emph{(relative) pseudo-volumetric mass flux}: 
\be p_0 := \int_{-d}^{\eta(x)} \sqrt{\varrho(x,y)}\left[u(x,y) -c \right] \, dy, \label{defp0} \ee
which is a (strictly negative) constant.  Likewise, we define 
\be p_i := \int_{-d}^{\eta_i(x)} \sqrt{\varrho(x,y)}\left[u(x,y) -c \right] \, dy, \label{defpi} \ee
so that $\{ y = \eta_i(x)\}$ coincides with the level set $\{ \psi = -p_i\}$.  

The conservation of mass \eqref{weakmass} implies that there exists a function $\rho:[p_0,0] \to \mathbb{R}^{+}$  such that 
\be \varrho(x,y) = \rho(-\psi(x,y)) \qquad \textrm{in } \Omega.  \label{defrho} \ee
We refer to $\rho$ as the \emph{streamline density function} and treat it as given.  For solitary waves, one determines $\rho$ by examining the flow  upstream or downstream.  

Conservation of energy can be expressed via Bernoulli's theorem, which states that the quantity
\be E := P + \frac{\varrho}{2}\left( (u-c)^2 + v^2\right) + g\varrho y, \label{defE} \ee
is constant along streamlines. This allows us to define the so-called \emph{Bernoulli function} $\beta:[0,|p_0|] \to \mathbb{R}$ by 
\be \frac{dE}{d\psi}(x,y) = -\beta(\psi(x,y)), \qquad \textrm{in } \Omega. 
\label{defbeta} \ee
The dynamics in the interior of each layer is then captured by Yih's equation \cite{yih1965dynamics}
\bse \label{stream function formulation} 
\be \Delta \psi - g y \rho^\prime(-\psi) + \beta(\psi) = 0 \qquad \textrm{in }  \Omega. \label{yih equation} \ee
Note that if in some layer $\Omega_i$ the flow is irrotational and $\rho_i$ is a constant, then \eqref{yih equation} reduces to Laplace's equation for $\psi$.  

Lastly, for the stream function formulation to be equivalent to the full Euler system, it is necessary that $E$ satisfies certain jump conditions on the free surfaces.  On the air--sea interface, $P = P_{\textrm{atm}}$ and hence \eqref{defE} gives immediately that 
\be |\nabla \psi|^2 + 2g \varrho_N (y + d) = Q, \qquad \textrm{on } \{ y = \eta(x) \}, \label{psi Bernoulli top} \ee 
where $Q = 2(E+\varrho d)|_{y=\eta}$.  
Likewise, for the interior interfaces we use the fact that $P$ is continuous in $\overline{\Omega}$ when evaluating $E$ from above and below on $\{ y = \eta_i(x) \}$.  This results in  
\be |\nabla \psi_{i+1}|^2 - |\nabla \psi_{i}|^2 + 2g (\rho_{i+1} - \rho_{i}) (\eta_i + d) =Q_i \qquad \textrm{on } \{ y = \eta_i(x) \}, \label{jump E} \ee
with $Q_i := 2( E_{i+1} + g \rho_{i+1} - E_i - g \rho_i)|_{y = \eta_i}$.  

For solitary waves, the form at infinity strongly constrains the global structure.  First, observe that \eqref{upstream condition} implies that \be  \nabla^\perp \psi = \sqrt{\varrho}(u-c,v) \to \sqrt{\rho}(\mathring{u}-c,0), ~ \eta(x) \to 0 \qquad \textrm{as } |x| \to \infty. \label{upstream psi} \ee
Let us focus on the simplest case where $\mathring{u} = 0$.  Then, letting $\mathring{h}(p)$ denote the limiting height above the bed of the streamline $\{ \psi = -p \}$, from \eqref{defpsi} we have 
 \begin{equation}\label{limiting h}
\mathring{h}(p) = \int_{p_0}^p \frac{1}{c \sqrt{\rho(s)}} \, ds.
\end{equation}
In fact, this couples the wave speed and density because we can evaluate the above identity on the air--sea interface to get 
\be \label{c determined by rho} d = \int_{p_0}^0 \frac{1}{c\sqrt{\rho(s)}} \, ds. \ee
Now,  it follows that
\begin{equation}\label{Q}
Q  = \varrho_N c^2 + 2g \varrho_N d,  \quad
Q_i  = \LC \varrho_{i+1} - {\varrho}_i \RC \left( c^2 + 2g \mathring{h}\right).
\end{equation}

Finally, taking the upstream limit of \eqref{defE}, we can actually compute $\beta$ directly in terms of $\mathring{h}$:
\be \beta(p) = \rho^\prime(p) \left[ \frac{1}{2} c^2 + g (\mathring{h}-d)\right].  \label{solitary beta} \ee
This is a well known fact about solitary stratified waves, but an elementary derivation is given in \cite{ChenWalsh2014continuity}.  
\ese

In summary, the stream function formulation of the problem is to find $\psi$ continuous in $\overline{\Omega}$ that satisfies Yih's equation \eqref{yih equation} in each $\Omega_i$, along with Bernoulli boundary condition \eqref{psi Bernoulli top} on the air--sea interface, the jump condition \eqref{jump E} on the internal interfaces, and the no stagnation condition \eqref{nostagnation}.  In addition, for a solitary wave, we must have that $\psi$ exhibits the asymptotic behavior \eqref{upstream psi}, which in turn implies that the constant $Q$, $Q_i$ are given as in \eqref{Q}.  

\subsection{Height function formulation}\label{height function subsection}
The stream function formulation has the advantage of being a scalar equation, but it is still posed in an a priori unknown domain.  For that reason, it is useful to consider changes of variables that map $\Omega$ to a fixed domain.  Of course, one should expect this to come at the cost of additional nonlinearity in the resulting PDE.  A particularly convenient choice is to adopt semi-Lagrangian variables 
$$ (x,y) \mapsto (q,p) := (x, -\psi).$$
These coordinates have a long history in the literature on existence of stratified water waves; their use dates back to Dubreil-Jacotin \cite{dubreil1937theoremes}.  Note that, because the internal interfaces and free surface are level sets of $\psi$, each layer $\Omega_i$ maps to a strip $\mathcal{R}_i := \{ (q,p) : p \in (p_{i-1}, p_{i} ) \}.$  Similarly, $\mathcal{R} := \bigcup_i \mathcal{R}$ is the image of $\Omega$.  

Let 
$$ h = h(q,p) := y +d,$$
which denotes the height above the bed of the point with $x$-coordinate $q$ that lies on the streamline $\{ \psi = -p\}$.  Notice that $h$ is continuous in $\overline{\mathcal{R}}$, though its derivatives will in general not be.  

One can show that Yih's equation \eqref{yih equation} is equivalent to the quasilinear elliptic problem 
\bse \label{heighteq general} 
\be 
 (1+h_q^2)h_{pp} + h_{qq}h_p^2 - 2h_q h_p h_{pq}  -g(h-d)\rho_p h_p^3  = -h_p^3 \beta(-p) \qquad \textrm{in } \mathcal{R}. \label{height equation} \ee
 The Bernoulli condition on the top \eqref{psi Bernoulli top} likewise becomes 
 \be 1+h_q^2 + h_p^2(2g\rho h- Q)  = 0 \qquad \textrm{on } \{ p = 0 \}, \label{h Bernoulli top} \ee
 while the jump conditions on the internal interfaces translate to 
 \be \frac{1+(\partial_q h_{i+1})^2}{ (\partial_p h_{i+1})^2} - \frac{1+(\partial_q h_{i})^2}{(\partial_p h_{i})^2} + 2 g(\rho_{i+1} - \rho_{i}) h = Q_i \qquad \textrm{on } \{ p = p_i\}. \label{h jump E} \ee
 Here we are adapting our previous convention and writing $h_i := h|_{\mathcal{R}_i}$.   Finally, simply from the definition it follows that
 \be h = 0 \qquad \textrm{on } \{p = p_0\}. \label{h bed cond} \ee
 \ese
 The derivation of this system can be found in many places (see, e.g., \cite{walsh2009stratified,walsh2013wind}).  
 
Restricting our attention to solitary waves with uniform velocity at infinity, \eqref{heighteq general} can be simplified by taking $\beta$ of the form \eqref{solitary beta}, and  $Q$, $Q_i$ as in \eqref{Q}.  For convenience, we record here the resulting system for a continuously stratified fluid (the layer-wise smooth case being a straightforward generalization):
\be\label{h eqn} \left \{ \begin{array}{lll}
(1+h_q^2)h_{pp} + h_{qq}h_p^2 - 2h_q h_p h_{pq} = -h_p^3 \rho_p\LB {1\over2}c^2 + g (\mathring{h} - h)  \RB, & \textrm{in } \mathcal{R},\\ 
& & \\
1+h_q^2 + h_p^2(2g\rho h- Q)  = 0, & \textrm{on }  \{p = 0\}, \\
h = 0, & \textrm{on } \{ p = p_0\}. \end{array} \right. \ee
Naturally, the definition of $\mathring{h}$  \eqref{limiting h} implies that we also have  
\be h \to \mathring{h} \qquad \textrm{as } |q| \to \infty. \label{asymptotic h} \ee  

\subsection{Reformulation using conformal coordinates}\label{subsec_conformal}

Recall that $\psi$ is given by \eqref{defpsi} along with the proscription that $\psi$ is continuous on $\overline{\Omega}$.  It follows that $\psi$ can be written  
\begin{equation*}
\psi(x, y) = -p_0 + \int^y_{-d} \sqrt{\rho(x, z)} \LB u(x, z) - c \RB\ dz.
\end{equation*}
From this expression, it is easy to confirm that $\psi \in C^{0,1}(\overline{\Omega})$. Observe also that Yih's equation \eqref{yih equation} implies that, if $\beta$ takes the form \eqref{solitary beta}, then in any layer where $\varrho_i$ is constant, the corresponding velocity field is irrotational and $\psi_i$ is harmonic.    

We may therefore define a harmonic function $\varphi$ with domain $\bigcup_i \Omega_i$ by taking $\varphi|_{\Omega_i}$ to be a harmonic conjugate of $\psi|_{\Omega_i}$. Thus $\nabla \varphi = \sqrt{\rho}(u-c,v)$, i.e., $\varphi$ is a pseudo velocity potential in each layer.  This only determines $\varphi|_{\Omega_i}$ up to a constant.  Anticipating an imminent change of variables, we choose to normalize so that 
\be \varphi(0,-d) = 0 \qquad \textrm{and} \qquad \{ (x,y) \in \overline{\Omega} : \varphi(x,y) = 0 \} \textrm{ is a connected set}. \label{defphi alt} \ee 
Note that there exists a unique continuous extension of $\varphi|_{\Omega_i}$ to $\overline{\Omega_i}$, but in general $\varphi$ will not be continuous over the interior interfaces.  This can be viewed as a consequence of the fact that the tangential velocity may jump between layers.   

As $\varphi$ and $\psi$ are harmonic conjugates in each $\Omega_i$, one can employ the hodograph transform: 
$$\mathscr{K}: (x,y) \mapsto \LC -\frac{1}{c}\varphi(x,y), \, -\frac{1}{c} \psi(x,y) \RC =: (t, s).$$
Then the image of each fluid layer $\Omega_i$ is a strip 
$$ \mathpzc{T}_i := \{ (t,s) \in \mathbb{R} \times (s_{i-1} , s_{i})\} $$
with $s_i := p_i/c$.  The fluid domain $\Omega$ is thus mapped to an $N$-slitted strip $\mathpzc{T} := \bigcup_i \mathpzc{T}_i. $

To keep the various formulations clear, designate the height above the bed in the $(t,s)$-coordinates by $H$: 
$$
H(t, s) = y + d \quad \textrm{for } (t, s) \in \mathpzc{T}.
$$
Therefore the $i$-th interface $\{y = \eta_i(x)\}$ can be parametrized by $t$:
\be\label{Yi}
Y_i(t) = H_i(t, s_i) - d.
\ee
In order to return to the physical variables $(x,y)$, it is also necessary to find a parameterization $x_i = x_i(t, s_i)$ in a way such that
\be\label{parametrization}
(x, \eta_i(x)) = (x_i(t, s_i), Y_i(t)).
\ee 

To reformulate the problem in the conformal coordinates, we begin collecting some simple change of variable calculations:  
\begin{equation}\label{chain rule}
\left\{\begin{split}
&\p_x = {\sqrt{\varrho_i}\over c}\LB (c-u)\p_{t} + v\p_{s} \RB, \\
&\p_y = {\sqrt{\varrho_i}\over c}\LB -v\p_{t} + (c-u)\p_{s} \RB,
\end{split}\right. 
\quad
\left\{\begin{split}
&H_{t} = y_{t} = -{cv \over \sqrt{\varrho_i} \LB (c-u)^2 + v^2 \RB} = -x_s,\\
&H_{s} = y_{s} = {c(c - u) \over \sqrt{\varrho_i} \LB (c-u)^2 + v^2 \RB} = x_t,
\end{split}\right. 
\quad\ \textrm{ in  } \mathpzc{T}_i.
\end{equation}
From this it also follows that
\begin{equation}\label{chain rule ts}
\left\{\begin{split}
&\p_{t} = H_{s}\p_x + H_{t} \p_y,\\
&\p_{s} = -H_{t} \p_x + H_{s} \p_y,
\end{split}\right.
\qquad
\left\{\begin{split}
v & = -{cH_t \over {\sqrt{\varrho_i}} \LC H_{t}^2 + H_{s}^2\RC},\\
c-u & = {cH_{s} \over {\sqrt{\varrho_i}} \LC H_{t}^2 + H_{s}^2 \RC},
\end{split}\right. 
\quad\ \textrm{ in  } \mathpzc{T}_i.
\end{equation}
Therefore the height function satisfies
\begin{equation}\label{conformal height eqn}
\left \{ \begin{array}{ll}
\Delta_{(t,s)} H  = 0, \qquad \qquad & \textrm{in } \mathpzc{T},\\
H = 0, & \textrm{on } s = s_0,\\
\displaystyle (H^2_t + H^2_s) \LC c^2 -2gH + 2gd \RC = {c^2 \over \varrho_N^2 } & \textrm{on } s = 0. \end{array} \right. \ee

For localized solitary waves as considered in \cite{amick1984semilinear,amick1986global,terkrikorov1963theorie,turner1981internal,turner1984variational}, we have the following asymptotics
\begin{equation}\label{asymptotics}
\left\{\begin{array}{ll}
\displaystyle \lim_{|t|\to \infty} \nabla_{(t,s)} H = \LC 0, {1\over \sqrt{\varrho_i}}\RC \quad \textrm{uniformly for } s \in (s_{i-1}, s_i).\\\\
\displaystyle \lim_{|t|\to \infty} H(t, s_i) = \mathring{H}(s_i),  \quad i = 1, \cdots N, 
\end{array}\right.
\end{equation} 
where $\mathring{H}(s) = \mathring{h}(cs)$. Moreover using \eqref{limiting h} to evaluate $\mathring{H}$ at $s_i$ yields
\be\label{limiting H values}
\mathring{H}(s_i) = \int_{p_0}^{p_i} \frac{1}{c \sqrt{\rho(\tau)}} \, d\tau = {1\over c} \sum^{i-1}_{j = 0} {p_{j+1} - p_j \over \sqrt{\varrho_{j+1}}} = \sum^{i-1}_{j = 0} {s_{j+1} - s_j \over \sqrt{\varrho_{j+1}}}.
\ee

\section{Determinability} \label{determinability section}
The objective of this section is to establish that the pressure reconstruction problem is mathematically well-defined.   That is, we prove in various regimes that the wave speed, streamline density function, Bernoulli function, and the trace of the pressure on the bed together uniquely determine a steady stratified water wave.   In particular, among all solitary waves with a given form at infinity, to each pressure trace there corresponds precisely one surface profile. 

The main insight underlying the argument is that, due to Bernoulli's law, the trace of the pressure on the bed gives Neumann data for the height function.  Since we already know that the height is $0$ on the bed, together this amounts to Cauchy data.  

To see this, first observe that, in the semi-Lagrangian coordinates, Bernoulli's law \eqref{defE} becomes
\be\label{Bernoulli in pq}
E = P + {1 + h^2_q \over 2 h^2_p} + g \rho (h - d).
\ee
Evaluating this on the flat bed and using the definition of the Bernoulli function $\beta$ \eqref{defbeta} leads to
\be\label{Bernoulli bottom general}
E|_{y = -d} = E|_{\psi = p_0} = E|_{\psi = 0} - \int^{p_0}_0 \beta(p)\ dp.
\ee
As before, the form of the flow in the far field \eqref{upstream condition} allows us to compute that the Bernoulli constant $E|_{\psi = 0}$ is
\be\label{Bernoulli top}
E|_{\psi = 0} = E|_{y = \eta} = P_{\textrm{atm}} + {1\over2} \rho(0) (c - \mathring{u}(0))^2.
\ee

Therefore, from \eqref{h bed cond} and \eqref{Bernoulli in pq}--\eqref{Bernoulli top}, we infer that
$$
P(q, p_0) + {1\over 2 h^2_p(q, p_0)} - g\rho(p_0)d = P_{\textrm{atm}} + {1\over2} \rho(0) (c - \mathring{u}(0))^2 - \int^{p_0}_0 \beta(p)\ dp.
$$
Finally, the absence of stagnation \eqref{nostagnation} implies $h_p > 0$, so we conclude 
\be\label{pressure transfer}
h_p(q, p_0) = \LB 2\LC P_{\textrm{atm}} + {1\over2} \rho(0) (c - \mathring{u}(0))^2 + g\rho(p_0)d - \int^{p_0}_0 \beta(p)\ dp -  P(q, p_0)\RC \RB^{-1/2}.
\ee
Thus the normal derivative of $h$ on $\{ p = p_0 \}$ is known as soon as the streamline density $\rho$,  Bernoulli function $\beta$,  limiting downstream/upstream horizontal velocity on the free surface, and the pressure data on the sea-bed are specified.

\subsection{Analytic setting}

We have seen that, treating the pressure trace as given,  \eqref{h bed cond} and \eqref{pressure transfer} amount to Cauchy data.  It is therefore useful to transform the elliptic height equation \eqref{heighteq general} into a first-order system that can be approached via Cauchy--Kowalevski. Define
\begin{equation}\label{newvar}
F := {h_q \over h_p}, \quad G := {1\over h_p}, \quad H := h.
\end{equation}
This is an abuse of notation in the sense that $H$ above is not the same as the one appearing in Section \ref{subsec_conformal}.  Nonetheless, for the purposes of the discussion in this subsection, there is little risk of confusion.  

From \eqref{nostagnation} and the identity  $1/h_p = \sqrt{\rho} (c - u)$ we know that these quantities are well-defined. We can then reformulate the height equation \eqref{heighteq general} as
\begin{equation}\label{1storder} \left\{ 
\begin{split}
& F_p = {F\over G} \LB {GF_q + FG_q + G \beta(-p) - Gg \rho_p (H - d) \over F^2 + G^2} \RB - {G_q \over G},\\
& G_p = {GF_q + FG_q +  G \beta(-p) - Gg \rho_p (H - d) \over F^2 + G^2},\\
& H_p = {1\over G},
\end{split} \right.
\end{equation}
with the following ``initial" conditions posed at $p = p_0$
\begin{equation*}
F(q, p_0) = H(q, p_0) = 0.
\end{equation*}
Here $G(\cdot, p_0)$ is uniquely determined by the pressure trace $P(\cdot, p_0)$ via \eqref{pressure transfer}. 
\begin{thm}\label{thm_determinability}
Let $h \in C^{2,\alpha}(\overline{\mathcal{R}})$ be a solution of the height equation \eqref{h eqn} representing a solitary water wave with a given stable real analytic streamline density stratification, a real analytic Bernoulli function, wave speed $c$, and  downstream/upstream horizontal velocity on the surface. Assume that the wave speed exceeds the horizontal fluid velocity $u < c$ through out the fluid. Then $h$ is real analytic and is uniquely determined by the pressure function $P$ on the flat bed.
\end{thm}
\begin{proof}
The Cauchy-Kowalevski theorem applied to system \eqref{1storder} implies the existence and uniqueness of $F, G, H$, which are real analytic in some open neighborhood $\mathcal{N}$ of the bed $\{p = p_0\}$. Hence, in view of \eqref{newvar}, $h, h_p, h_q$ are also real analytic in $\mathcal{N}$. On the other hand, under the regularity hypotheses above,  all solutions of the height equation \eqref{heighteq general} are real analytic in $\overline{\mathcal{R}}$ (cf. \cite[Theorem 5.1]{wangregularity2013}).   Therefore the unique continuation of real analytic functions ensures the uniqueness of the solution $h$ to \eqref{h eqn} with Neumann data on the bottom given by the pressure from \eqref{pressure transfer}.
\end{proof}

%

In contrast to the case of solitary waves, a fairly robust existence theory for large-amplitude periodic steady stratified waves has been developed in recent years (cf. \cite{walsh2009stratified,escher2011stratified,henry2013local,henry2013global,walsh2014local,walsh2014global}). Since the Cauchy-Kowalevski theory is local, it can be applied just as well on a periodic domain.  Thus, in order to adapt the argument of Theorem \ref{thm_determinability} to periodic traveling waves, it remains only to show that the pressure trace on the bed furnishes Neumann data for $h$ there.  Unfortunately, because this setting lacks an upstream state, we will need some additional information.  

More precisely, observe that by combining \eqref{Bernoulli in pq} with \eqref{Bernoulli bottom general} for a general Bernoulli function $\beta$, we obtain the identity 
\be h_p(q,p_0) = \left[ 2 \left(  E|_{\psi = 0} - \int_0^{p_0} \beta(p) \, dp + g \rho(p_0) d - P(q, p_0)  \right) \right]^{-1/2}. \label{hp from P periodic} \ee
So long as we know $E|_{\psi = 0}$, we can proceed as before.  This leads to the following result.  
\begin{cor}\label{cor_periodic}
Let $h \in C^{2,\alpha}(\overline{\mathcal{R}})$ be a periodic solution of the height equation \eqref{heighteq general} for a stable  streamline density function $\rho \in C^\omega([p_0,0])$,  Bernoulli function $\beta \in C^\omega([0,|p_0|])$, and with wave speed $c$. Then $h \in C^\omega(\overline{\mathcal{R}})$.  Moreover, $h$ is uniquely determined by the trace of the pressure $P$ on the bed and the value of $E$ on the free surface. 
\end{cor}
\begin{rem}
There are many situations in which it is possible to find $E|_{\psi=0}$ from other information about the flow.  For instance,  any one of the following is enough to determine $E|_{\psi=0}$:
\begin{itemize}
\item[(i)] the value of $u$ at the crest and the height of the crest; or
\item[(ii)] the value of $u$ at the trough and the height of the trough; or
\item[(iii)] the average value of the kinetic energy ${\rho(0)\over2} |(u-c,v)|^2$ on the surface.
\end{itemize}
For (i), we just evaluate Bernoulli's law at the crest (where $v = 0$).  Likewise, for (ii), we evaluate Bernoulli at the trough.  For (iii), we
take the mean value on the surface and use the fact that $\eta$ has mean $0$. 
\end{rem}

These determinability results can be easily extended to layer-wise smooth flows.  Note that the interior interfaces are streamlines, hence they correspond to straight lines $\{ p = p_i\}$ in the semi-Lagrangian variables.  The first-order system \eqref{1storder} will hold in each strip $\mathcal{R}_i$. Beginning in the strip directly above the bed, we can argue as in Theorem \ref{thm_determinability} and Corollary  \ref{cor_periodic} to show that  $h \in C^\omega(\overline{\mathcal{R}_1})$ and is determined there uniquely by the prescribed data.  $h$ is continuous over the interface and indeed $q \mapsto h(q, \cdot)$ is real analytic.  On the other hand, the jump condition \eqref{h jump E} allows us to determine $\partial_p h_2(\cdot, p_1)$.  Thus we have a full set of Cauchy data for the first-order problem in $\mathcal{R}_2$.  Applying Cauchy-Kowalevski once more allows us to conclude $h \in C^\omega(\overline{\mathcal{R}_2})$.  Iterating this procedure, we get the following result.  
\begin{cor}\label{cor_layered}
Let $h \in C^{0,\alpha}(\overline{\mathcal{R}})$ be a solitary wave solution of \eqref{heighteq general} with wave speed $c$ such that $h_i \in C^{2,\alpha}(\overline{\mathcal{R}_i})$.  Suppose that for each $i$, $\rho_i \in C^\omega([p_{i-1}, p_i])$ and $\beta_i \in C^\omega([|p_i|, |p_{i-1}|])$. Then $h$ is real analytic in each layer $\mathcal{R}_i$, and it is uniquely determined by the trace of the pressure on the bed. \end{cor}

\subsection{Sobolev setting} In this section, we considerably strengthen the above results by relaxing the analyticity requirement to a more physically reasonable degree of smoothness.  Our argument is built off of the following result on the strong unique continuation for a general elliptic equation set in the plane.
\begin{thm}\cite[Theorem 1]{Alessandrini2012}\label{thm_strong continuation}
Given a bounded connected open set $\Omega \subset \R^2$. Let $u\in W^{1,2}(\Omega)$ be a weak solution to the elliptic equation
$$
Lu: = (a^{ij} u_{x_i} + a^j u)_{x_j} + b^i u_{x_i} + cu  = 0  \quad \textrm{in } \Omega,
$$
where 
\begin{enumerate}
\item[(i)] $L$ is uniformly elliptic in the sense that there exists a $K>0$ such that 
$$
0< K^{-1} |\zeta|^2 \leq a^{ij}(x) \zeta_i \zeta_j \leq K |\zeta|^2, \quad \textrm{ for all } \zeta\in \R^2, \textrm{and for a.e. } x\in \Omega; and 
$$
\item[(ii)] $a^{ij} \in L^\infty(\Omega)$, $a^i,\ b^i \in L^q(\Omega)$, $c\in L^{q/2}(\Omega)$ for some $q>2$.
\end{enumerate}
Then the operator $L$ has the strong unique continuation property in the sense that if a solution $u$ has a zero of infinite order at a point $x_0 \in \Omega$, i.e., there exists $\delta>0$ such that for every integer $N\geq0$
$$
\int_{B_R(x_0)} |u|^2 \ dx \leq c_N r^N, \quad \textrm{ for all } R<\delta,
$$
then $u\equiv 0$ in $\Omega$.
\end{thm} 

With this result in place, we can now state and prove our main theorem on determinability.  
\begin{thm}\label{thm_determinability normal derivative}
Let $r>4$ and $h \in W^{3,r}(\mathcal{R})$ be a solution of \eqref{h eqn} representing a solitary water wave with a stable streamline density stratification $\rho \in W^{2,r}([p_0, 0])$ and a Bernoulli function $\beta \in W^{1,r}([0, |p_0|])$, and such that the wave speed exceeds the horizontal fluid velocity. Then $h$ is uniquely determined by the normal derivative $h_p$ on the flat bed. 
\end{thm}
\begin{proof}
Let $h, \tilde{h}\in W^{3,r}(\mathcal{R})$ be two solutions to \eqref{h eqn} with identical Neumann data on the bed:
$$ h_p(\cdot, p_0) = \tilde{h}_p(\cdot, p_0) \qquad \textrm{on } \mathbb{R}.$$
The difference $u := h - \tilde{h}$  satisfies an elliptic equation (see also \cite{walsh2009symmetry})
\begin{equation*}
\mathscr{L}u := (a^{ij} u_{x_i})_{x_j} + b^i u_{x_i} + cu  = 0 \quad \textrm{ in } \mathcal{R},
\end{equation*}
where
\begin{align*}
a^{11} & := 1+ h^2_q, \quad a^{12} = a^{21} = - h_ph_q, \quad a^{22} = h^2_p,\\
b^1 & := \tilde{h}_{pp} (h_q + \tilde{h}_q) - 2 h_p \tilde{h}_{pq} - h_p h_{pq} + h_{q} h_{pp},\\
b^2 & := \tilde{h}_{qq}(h_p + \tilde{h}_p) - 2 \tilde{h}_q \tilde{h}_{pq} + \beta(-p) (h^2_p + h_p \tilde{h}_p + \tilde{h}^2_p) - g \rho_p (\tilde{h} - d(\tilde{h})) (h^2_p + h_p \tilde{h}_p + \tilde{h}^2_p) \\
& \quad - h_q h_{pq} + h_p h_{qq}\\
c & := - g \rho_p h^3_p,
\end{align*}
and 
$$
u = u_p = 0, \quad \textrm{on } \{p = p_0\}.
$$
The above boundary condition guarantees that $u$ can be extended to a $W^{1,r}$ function defined on the half-plane $\{ p< 0\}$ that vanishes on $\{ p < p_0 \}$. Also, since $h, \tilde{h}\in W^{3,r}(\mathcal{R})$, from the trace theorem and fractional Sobolev embedding theorem (cf., e.g.,  \cite[Theorem 8.2]{FractionalSobolev}) we know that 
$$\nabla^2 h, \nabla^2 \tilde{h} \in W^{1/2,r}(\{p=0\} \cup \{ p = p_0\}) \subset C^{0, {r-2\over 2r}}(\{p=0\} \cup \{ p = p_0\}).$$ 
Similarly, we can make sense of $\rho_p$ and $\beta$ on $\{ p = 0\}$ and $\{p = p_0 \}$. This way, the functions $h, \tilde{h}, \rho$ and $\beta$ can be extended outside $\mathcal{R}$ while maintaining their regularity and ensuring $h_p > 0$ in the strong sense. Therefore, from the divergence structure of the operator $\mathscr{L}$, we see that the extended $u$ also satisfies $\mathscr{L}u = 0$. This, combined with Sobolev embedding, implies that the coefficients in $\mathscr{L}$ satisfy the the conditions of Theorem \ref{thm_strong continuation}, and hence $\mathscr{L}$ has the strong unique continuation property. Note that Theorem \ref{thm_strong continuation} is a local result. Thus we may apply it on any finite strip of the fluid domain to show that the restriction of $u$ there vanishes identically.  It follows that $u \equiv 0$ on the entire strip, completing the proof of theorem.
\end{proof}
As discussed in the previous subsection, $h_p(\cdot, p_0)$ is determined by $P({\cdot, p_0})$ through \eqref{pressure transfer}. Thus Theorem \ref{thm_determinability} can be reinterpreted as a statement about unique determinability of $h$ from the pressure trace.  
\begin{thm}\label{thm_determinability general}
Let $r>2$ and $h \in W^{3,r}(\mathcal{R})$ be a solution of \eqref{h eqn} representing a solitary water wave with a stable streamline density stratification $\rho \in W^{2,r}([p_0, 0])$, a Bernoulli function $\beta \in W^{1,r}([p_0, 0])$.  Fix the wave speed $c$, and the downstream/upstream horizontal velocity on the surface. Then $h$ is uniquely determined by the trace of the pressure $P$ on the bed. In particular, the pressure function on the flat bed uniquely determines the free surface profile.
\end{thm}
\begin{rem} \label{henry remark}
The same argument applies to the regime considered by Henry \cite{henry2013pressure}, namely constant density flow with a general vorticity distribution. In that case, the height equation is  
\be\label{rotational height} \left \{ \begin{array}{lll}
(1+h_q^2)h_{pp} + h_{qq}h_p^2 - 2h_q h_p h_{pq} = -h_p^3 \gamma(-p), & p_0 < p < 0, \\
1+h_q^2 + h_p^2(2g h- Q)  = 0, & p = 0, \\
h = 0, & p = p_0, \end{array} \right. \ee
where $\gamma$ is the vorticity function. Clearly $\gamma$ plays a nearly identical role as $\beta$ does in \eqref{h eqn}, and so the arguments above go through with only minor alteration.
\end{rem}

\section{Reconstruction for layer-wise constant density irrotational waves} \label{streamline recovery section}
Up to this point, we have proved that, in a number of physical regimes, the pressure trace on the bed uniquely determines a traveling wave in a stratified water.  The goal of the present section is to develop an explicit iteration scheme that reconstructs the internal interfaces of a layer-wise constant density and irrotational solitary wave given the pressure data.

\subsection{Surface reconstruction from pressure}\label{subsec_reconstruction} In Section \ref{subsec_conformal}, we introduced the hodograph transform and the reformulation of the height equation in $(t,s)$-coordinates.  Recall that the density is constant within each layer and we denote the height function as $H = H(t,s)$ when working in the conformal variables. Reading off the Bernoulli law in the bottom layer from far-field data yields
\begin{equation}\label{Bernoulli bottom}
{P_{\textrm{bot}} + {c^2 \over 2 \LB(\p_tH_1)^2 + (\p_sH_1)^2\RB} - g\varrho_1 d = P_{\textrm{atm}} + \int^0_{p_0} g \rho(p) \ dp + {\varrho_1 c^2 \over 2} - g \varrho_1 d,}
\end{equation}
where the second term on the right-hand side comes from inverting the streamline coordinates from the Eulerian one at infinity. 
The boundary condition on $H_1$, \eqref{chain rule}, and the fact that $H_s > 0$ then give a formula for $\p_sH_1(\cdot, s_0)$
\begin{equation}\label{H_s on bottom}
\p_sH_1(t, s_0) = {c \over \sqrt{\varrho_1} \sqrt{c^2 - 2 \mathfrak{p}(t)}},
\end{equation}
where
\begin{equation}\label{bottom pressure} 
\mathfrak{p}(t) := \dfrac{P_{\textrm{bot}}({{x_0}(t)}) - P_{\textrm{atm}}  }{\varrho_1} - {1\over \varrho_1}\int^0_{p_0} g \rho(p) \ dp.
\end{equation}
Recall that $x_0$ is the inverse of $x \mapsto \varphi_1(x, -d)$.

The surface reconstruction problem can then be stated as follows:  From the form of the wave at infinity as described by the data $\mathring{H}$, $d$, $\rho$, and $c$, and the pressure trace on the bed as described by $\mathfrak{p}$, find a parameterization for each of the internal interfaces as well as the air--sea interface. 

\begin{rem} \label{p(t) from Pbot remark}  $\mathfrak{p}$ is a more convenient quantity to work with mathematically, but will not be directly observable in most physical scenarios.  However, it can be determined uniquely from $P_{\textrm{bot}}$, the trace of the pressure on the bed in the original variables.   This follows from several observations.  First, note that the vertical velocity vanishes on the bed, and thus from Bernoulli's law  we find that:
$$ \varphi_x(x,-d) = -\left[ P_{\textrm{atm}} - P_{\textrm{bot}}(x) +   \int^0_{p_0} g \rho(p) \ dp + {\varrho_1 c^2 \over 2} \right]^{1/2}.$$
Here we are using the fact that the absence of stagnation implies  $\varphi_x(0,-d) = \sqrt{\rho}_1(u-c) < 0$, which selects the negative branch of the square root.  Thus $x_0$ is characterized as the solution to the ODE initial value problem
$$ \left\{ \begin{split}  \dot{x}_0(t) &= \frac{1}{\varphi_x(x_0(t))} = - \left[ P_{\textrm{atm}} - P_{\textrm{bot}}(x_0(t)) +  \int^0_{p_0} g \rho(p) \ dp + {\varrho_1 c^2 \over 2} \right]^{-1/2} \\
x_0(0) &= 0. \end{split} \right.$$
The condition $x_0(0) = 0$ follows from our choice to take $\varphi(0,-d) = 0$.  Once $x_0$ is known, $\mathfrak{p}$ is determined uniquely by \eqref{bottom pressure}. 
\end{rem}

It is important to note that one cannot hope to choose $\mathfrak{p}$ arbitrarily and have it correspond to the pressure trace of an actual traveling wave.   We will require a posteriori information about $H$; that is, we need to make use of the fact that it is known that we are observing a solution, and that solutions of the system have a certain regularity and decay at infinity.   Otherwise, the ill-posedness of the problem cannot be avoided.

Fix an $\alpha \in (0,1)$.  We will assume that the wave to be reconstructed lies in the set 
\be \label{def H set} \begin{split} \mathscr{H}  = \mathscr{H}(\rho, \mathring{H}, c, d)  &:= \{ H \in C^{0,\alpha}(\overline{\mathpzc{T}}) :  H - \mathring{H}, H_t \textrm{ exponentially localized in $t$}, \\ 
& \qquad\qquad H \textrm{ a solution of \eqref{conformal height eqn}} \}.\end{split} \ee
The existence theory of Turner \cite{turner1981internal,turner1984variational} shows that $\mathscr{H}$ is non-empty if the wave speed is supercritical (cf. Section \ref{continuous section}).  In  \cite{amick1986global}, moreover, Amick and Turner prove that it includes waves of large amplitude in the case of a two-layer system.  
 
With this degree of regularity, our procedure holds in the strong sense; with less decay at infinity, we obtain only a distributional representation for the wave.   Of course, since $H \in \mathscr{H}$ is harmonic in each layer, it is $C^\infty(\overline{\mathpzc{T}_i})$ for $i = 1, \ldots, N$.   In fact, the structure of the equation and regularity theory for transmission boundary conditions (cf., e.g., \cite[Chapter 16]{ladyzhenskaya1968linear}) imply that $\p_t H \in C^{0,\alpha}(\overline{\mathpzc{T}})$.



\subsection{Jump condition in conformal coordinate system}\label{subsec_jump}
One of the main technical difficulties for multi-layered flows  is that the conformal coordinates behave badly on the image of the interfaces.  Suppose that $F \in C^0(\overline{\Omega})$ and let $f := F \circ \mathscr{K}^{-1} $.  Then, in general, $f \not\in C^0(\overline{\mathpzc{T}})$ since
$$ f_{i+1}(t,s_i) \neq f_{i}(t,s_i).$$
In other words, while each $\mathscr{K}_i$ is a conformal mapping, $\mathscr{K}$ itself is not a homeomorphism.  Of course, this stems from the fact that $\varphi$ may have a jump discontinuity over the interface.  Recall that we chose in \eqref{defphi alt} to normalize $\varphi$ so that $\varphi^{-1}(\{0\})$ is a connected set.  In terms of the hodograph transform, this becomes 
$$ (\mathscr{K}_{i+1})^{-1}(0, s_{i}) =( \mathscr{K}_{i})^{-1}(0, s_{i}), \qquad \textrm{for all } i = 1, \ldots, N-1,$$
which ensures that, at the very least, functions like $f$ above are continuous along the $s$-axis.   

Given this issue, it is necessary to examine the jump conditions on the internal interfaces very carefully.  Consider the situation on the $i$-th interior interface $\{ s = s_{i} \}$, for some $1 \leq i \leq N-1$.  Let 
\be\label{Xi}
X_{i+} := x_{i+1}(\cdot,s_{i}), \qquad X_{i-} := x_{i}(\cdot, s_{i}),
\ee
where $x_i(t,s)$ is the restriction of the Eulerian $x$-coordinate on $\mathpzc{T}_i$. Likewise define
$$\Phi_{i+}(x) := \varphi_{i+1}(x, \eta_{i}(x)), \qquad \Phi_{i-}(x) := \varphi_{i}(x, \eta_{i}(x)).$$
This simply says that $X_{i\pm}$ and $\Phi_{i\pm}$ are inverse to each other, i.e., 
\be 
X_{i\pm} = X_{i\mp}(\Phi_{i\mp}(X_{i\pm})), \qquad t = \Phi_{i\pm}(X_{i\pm}(t)).  
\label{reparameterization identities} \ee

The relation \eqref{reparameterization identities} allows us to translate from the parameterization of the interface by the coordinates in $\mathpzc{T}_{i+1}$, to that given by the coordinates in $\mathpzc{T}_{i}$:  define
\be \label{def mathfrakt} 
\tau_{i\pm}(t) := \Phi_{i\pm}(X_{i\mp}(t)), \ee
then 
\be\label{translation}
(\mathscr{K}_{i+1})^{-1}(t, s_{i}) =( \mathscr{K}_{i})^{-1}(\tau_{i-}(t), s_{i}), \qquad (\mathscr{K}_{i+1})^{-1}(\tau_{i+}(t), s_{i}) = (\mathscr{K}_{i})^{-1}(t, s_{i}).
\ee
We also mention that \eqref{defphi alt} implies 
\be X_{i+}(0) = X_{i-}(0) \qquad \textrm{and} \qquad \tau_{i-}(0) = 0. \label{alt initial cond} \ee

Suppose that $f$ is a function defined on $\overline{\mathpzc{T}_i} \cup \overline{\mathpzc{T}_{i+1}}$ and let $F = f \circ \mathscr{K}$ be its expression in physical variables.  Then $f_{i+1}(t,s_{i})$ corresponds to the value of $F_{i+1}$ at the Eulerian point $(x,\eta_i(x))$ where $\Phi_{i+}(x) = t$.  Chasing definitions, we see that this is the same point at which $\Phi_{i-}(x) = \tau_{i-}(t)$.  Therefore, the jump in $F$ over the $i$-th interface corresponds to the following difference in the conformal variables:
\be\label{newjump} \jump{f}_i := f_{i+1}(\cdot, s_{i}) - f_{i}(\tau_{i-}(\cdot), s_{i}).\ee
Hence the jump condition \eqref{jump E}  becomes  
\be \displaystyle \jump{ {c^2 \over  {H^2_t + H^2_s}}}_i + 2g \jump{\varrho (H - \mathring{H})}_i  = \jump{\varrho}_i c^2, \qquad  \textrm{on } \{s = s_{i}\}, \ i = 1, \ldots, N-1. \label{Bernoulli jump conformal} \ee

\subsection{Solution of the Cauchy problem}

Our procedure begins with the pressure data on the ocean bed, then reconstructs the flow in each layer moving from the lowest layer upward. At each step, this involves solving a non-standard elliptic problem:  $H$ is harmonic in each strip $\mathpzc{T}_i$, and we will have both the Dirichlet and Neumann data on the lower boundary.  This is essentially the Cauchy problem for Laplace's equation (cf. Section \ref{determinability section}), which is generally ill-posed.  In the present case, however, we are not concerned with existence theory --- we know the wave exists because we have measured it --- and the asymptotic decay properties \eqref{asymptotics} will allow us to make sense of the solutions as functions.  The ill-posedness issue reemerges if the pressure data is noisy, but we forestall that discussion until Section \ref{pressure noise section}.

The process of solving the Cauchy problem in each strip can be abstracted using nonlocal operators.  In fact, these operators are nothing but  Fourier multipliers whose symbols can be found via elementary separation of variables.  Explicitly, for each $i = 1, \ldots, N$, we define an operator $\mathcal{C}_i$ by
\be\label{DN operator}
\mathcal{C}_i (\phi_1, \phi_2) := (H|_{s = s_i}, H_s|_{s = s_i}),
\ee
where $H$ is the unique solution of the Cauchy problem
\be\label{DN problem}
\left\{\begin{array}{ll}
\Delta H = 0, \qquad & \textrm{in } \mathpzc{T}_i\\
H = \phi_1, \ \ H_s = \phi_2, & \textrm{on } \{s = s_{i-1}\}. 
\end{array}
\right.
\ee
That is, $\mathcal{C}_i$ takes the Cauchy data on the lower boundary of the $(i-1)$-st layer $\{ s = s_{i-1} \}$, solves the PDE, and then evaluates the Cauchy data for that solution on the upper boundary $\{ s = s_i \}$.  For now, we do not  specify precisely the domain and codomain of $\mathcal{C}_i$, but rather work formally.  

Taking the Fourier transform of \eqref{DN problem} in the $t$-direction, denoting by $\xi$ the Fourier variable, we see immediately that $\widehat{H} :=\mathcal{F}H$ must take the form 
\begin{equation}\label{solution to DN problem}
\widehat{H}(\xi, s) = \widehat{\phi}_1(\xi) \cosh\LB (s-s_{i-1})\xi \RB + {\widehat{\phi}_2(\xi) \over \xi} \sinh\LB (s-s_{i-1})\xi \RB.
\end{equation}
Hence
\begin{equation}\label{solution to DN}
\mathcal{F}[{\mathcal{C}_i (\phi_1, \phi_2)}](\xi) = \begin{pmatrix} \widehat{\phi}_1(\xi) \cosh\LB (s_i-s_{i-1})\xi \RB + {\dfrac{\widehat{\phi}_2(\xi)}{\xi}} \sinh\LB (s_i-s_{i-1})\xi \RB \\ \xi\widehat{\phi}_1(\xi) \sinh\LB (s_i-s_{i-1})\xi \RB + {\widehat{\phi}_2(\xi) } \cosh\LB (s_i-s_{i-1})\xi \RB \end{pmatrix}^T.
\end{equation}


\subsection{Reconstruction of the first internal interface}\label{subsec_firstlayer}

As a model computation, we derive a reconstruction formula for the first interface $\{y = \eta_1(x)\}$ given $P_{\textrm{bot}}(\cdot, s_0)$.  From \eqref{H_s on bottom} we recover the normal derivative of $H_1$ on the bed
\begin{equation}\label{p_s H bottom}
\p_sH_1(t, s_0) = {c \over \sqrt{\varrho_1} \sqrt{c^2 - 2 \mathfrak{p}(t)}}.
\end{equation}
Note that both sides above are bounded functions of $t$ and hence can be regarded as tempered distributions. Performing a  Fourier transform in $t$ yields 
$$\widehat{\p_sH_1}(\xi, s_0) = \mathcal{F}\LCB {c \over \sqrt{\varrho_1} \sqrt{c^2 - 2 \mathfrak{p}}} \RCB(\xi) $$
in the sense of distributions. Applying \eqref{solution to DN}, we find that
\begin{align}
\widehat{H_1}(\xi, s) & = { \sinh \LB (s-s_0)\xi \RB \over \xi} \mathcal{F}\LCB {c \over \sqrt{\varrho_1} \sqrt{c^2 - 2 \mathfrak{p}}} \RCB(\xi), \quad s\in [s_0, s_1], \label{H in 1st strip} \\
\widehat{\p_sH_1}(\xi, s_1) & = \cosh \LB (s_1 - s_0) \xi \RB \mathcal{F}\LCB {c \over \sqrt{\varrho_1} \sqrt{c^2 - 2 \mathfrak{p}}} \RCB(\xi). \label{H_s on s_1}
\end{align}

From the definitions of $x_i, H_i$, and $\mathscr{K}$, we know that the inverse of $\mathscr{K}|_{\overline{\mathpzc{T}_1}}$ is given by 
$$
x = x_1(t,s), \quad y = H_1(t,s) - d.
$$
This together with \eqref{Xi} and \eqref{Yi} yields
\begin{equation}\label{1st layer}
(x, \eta_1(x)) = (x_1(t,s_1), H_1(t, s_1) - d) = (X_{1-}(t), Y_1(t)).
\ee
Note that there is no reason to add a $\pm$ to the subscript for $Y_1$ because $s \mapsto y(\cdot, s)$, unlike $s \mapsto x(\cdot, s)$, is continuous.  

From \eqref{chain rule} we have the identity $\p_sH_1 = \partial_t x_1$. Therefore, \eqref{H in 1st strip} and \eqref{H_s on s_1} give 
\begin{align*}
\mathcal{F}\{Y_1\} (\xi) & = { \sinh \LB (s_1-s_0)\xi \RB \over \xi} \mathcal{F}\LCB {c \over \sqrt{\varrho_1} \sqrt{c^2 - 2 \mathfrak{p}}} \RCB(\xi) - 2\pi d \delta(\xi),\\
\mathcal{F}\{\dot{X}_{1-}\} (\xi) & = \cosh \LB (s_1 - s_0) \xi \RB \mathcal{F}\LCB {c \over \sqrt{\varrho_1} \sqrt{c^2 - 2 \mathfrak{p}}} \RCB(\xi),
\end{align*}
since $\mathcal{F}\{1\}(\xi) = 2\pi \delta(\xi)$. Using some elementary identities, we can simplify this further to obtain
\begin{align}
\mathcal{F}\LCB Y_1 - \LC {s_1 - s_0 \over \sqrt{\varrho_1}} - d \RC \RCB (\xi) & = { \sinh \LB (s_1-s_0)\xi \RB \over \xi} \mathcal{F}\LCB  {1\over \sqrt{\varrho_1}} \LC {c \over \sqrt{c^2 - 2 \mathfrak{p}}} - 1 \RC \RCB(\xi), \label{Fourier H_1}\\
\mathcal{F}\left\{\dot{X}_{1-} - {1\over \sqrt{\varrho_1}}\right\} (\xi) & = \cosh \LB (s_1 - s_0) \xi \RB \mathcal{F}\LCB {1\over \sqrt{\varrho_1}} \LC {c \over \sqrt{c^2 - 2 \mathfrak{p}}} - 1 \RC \RCB(\xi) \label{Fourier theta_1}.
\end{align}
Both \eqref{Fourier H_1} and \eqref{Fourier theta_1} hold in the sense of (tempered) distributions. Ideally, we would like them to hold in the strong sense, and to show this we use the idea of Constantin \cite{constantin2012pressure}. From \eqref{chain rule} and \eqref{H_s on bottom} we know that 
\begin{equation}\label{decay p}
[c - u(x, -d)]^2 = c^2 - 2\mathfrak{p}(t).
\end{equation}
On the other hand, we are working with solutions in the set $\mathscr{H}$, which means that $u(x,-d)$, $\nabla u(x, -d)$ decay exponentially fast as $|x| \to \infty$. So, by combining \eqref{decay p} and \eqref{chain rule ts}, we may conclude that $\mathfrak{p}$ and its derivative are square-integrable. In particular, $\mathfrak{p}(t) \to 0$ as $|t| \to \infty$. Moreover, writing
\begin{equation*}
{c \over \sqrt{c^2 - 2 \mathfrak{p}(t)}} - 1 =  {2\mathfrak{p} \over \sqrt{c^2 - 2 \mathfrak{p}(t)} \LC c + \sqrt{c^2 - 2 \mathfrak{p}(t)} \RC},
\end{equation*}
we see that that the left-hand side is likewise in $L^2$. 

Observe from \eqref{chain rule} and \eqref{Xi} that $\dot{X}_{1-} = \p_sH_1(\cdot, s_1)$, and therefore the terms on the left-hand side of \eqref{Fourier H_1} and \eqref{Fourier theta_1} are Fourier transforms of square integrable functions. This guarantees that \eqref{Fourier H_1} and \eqref{Fourier theta_1} are indeed equalities for functions. Putting together \eqref{1st layer}--\eqref{Fourier theta_1}, we obtain a parameterization of the first internal interface
\begin{align}
& X_{1-}(t) = {t\over \sqrt{\varrho_1}} + \int^{t}_{-\infty} \mathcal{F}^{-1} \LCB \cosh\LB (s_1 - s_0) \xi \RB \mathcal{F}\LB {1\over \sqrt{\varrho_1}} \LC {c \over \sqrt{c^2 - 2 \mathfrak{p}}} - 1 \RC \RB(\xi) \RCB (t') \ dt', \label{x 1st layer} \\
& Y_1(t) = \mathcal{F}^{-1}\LCB { \sinh \LB (s_1-s_0)\xi \RB \over \xi} \mathcal{F}\LB  {1\over \sqrt{\varrho_1}} \LC {c \over \sqrt{c^2 - 2 \mathfrak{p}}} - 1 \RC \RB(\xi) \RCB (t) + {s_1 - s_0 \over \sqrt{\varrho_1}} - d. \label{eta 1st layer}
\end{align}

\subsection{Reparameterization}\label{subsec repara}
Now suppose we have obtained the parametrization of the $i$-th interface $(X_{i-}(t), Y_i(t))$. In order to iterate the procedure given above, we will need to compute the function $\tau_{i-}$ that allows us to deal with the singularity in the coordinates on the interface.   It is simplest to do this working in the semi-Lagrangian variables.  Recall that the jump condition on $\{ s = s_{i} \}$ (which is also $\{ p = p_{i} \}$) is expressed as follows
\begin{equation*}
-{1+ (\p_q h_{i+1})^2 \over 2(\p_p h_{i+1})^2} + {1+ (\p_q h_{i})^2 \over 2(\p_p h_{i})^2} - g (\varrho_{i+1} h_{i+1} - \varrho_i h_i) + {Q_i \over 2} = 0 \qquad \textrm{on } \{ p = p_{i} \},
\end{equation*}
where 
\begin{equation*}
Q_i := 2\LB (E_{i+1} - E_i) + g(\varrho_{i+1} - \varrho_i) d \RB.
\end{equation*}
Using the fact that $h$ and $h_q$ are both continuous across the interface we can solve for $\partial_ph_{i+1}$
\be\label{jump h_p}
\partial_p h_{i+1} = \sqrt{1+ (\partial_q h_{i})^2} \LB Q_i - 2g (\varrho_{i+1} - \varrho_i) h_{i} + {1+ (\partial_q h_{i})^2 \over (\partial_p h_{i})^2} \RB^{-1/2} \qquad \textrm{on } \{ p = p_{i} \}.
%
\ee

Recall that we have the change of variable relations
\begin{equation*}
{1\over h_p} = -\psi_y = \varphi_x, \qquad {h_q\over h_p} = \psi_x = \varphi_y.
\end{equation*}
Thus, from \eqref{jump h_p} we can recover the full gradient $\nabla \varphi$ at each point on the $i$-th interface:
\begin{align*}
 \partial_x\varphi_{i+1}(x, \eta_{i}(x)) &= {1\over \partial_p h_{i+1}(x, p_{i})}  = \LB{1+ (\partial_q h_{i})^2}\RB^{-1/2} \left.\LB Q_i - 2g (\varrho_{i+1} - \varrho_i) h_{i} + {1+ (\partial_q h_{i})^2 \over (\partial_p h_{i})^2} \RB^{1/2}\right|_{(x, p_{i})},\\
 \partial_y\varphi_{i+1}(x, \eta_{i}(x)) &= {\partial_q h_{i+1}(x, p_{i})\over \partial_ph_{i+1}(x, p_{i})} = {\partial_q h_{i} \over \sqrt{1+ (\partial_q h_{i})^2}} \left.\LB Q_i - 2g (\varrho_{i+1} - \varrho_i) h_{i} + {1+ (\partial_q h_{i})^2 \over (\partial_p h_{i})^2} \RB^{1/2}\right|_{(x, p_{i})}.
\end{align*}


Let us now compute $\tau_{i-}$.  By differentiating \eqref{def mathfrakt} and using the identity 
$$\Phi_{i-}(X_{i+}) = \varphi_{i} (X_{i+}, \eta_{i}(X_{i+})),	$$ 
we obtain
\be\label{deriv t_i-}
\begin{split}
\dot{\tau}_{i-}(t) & = \dot{X}_{i+} (1, \eta'_i) \cdot \nabla \varphi_{i}(X_{i+}, \eta_i)\\
& = \dot{X}_{i+}\left. {1 + (\partial_q h_{i})^2 \over \partial_p h_{i}} \right|_{(X_{i+}, p_{i})} = {c\dot{X}_{i+} (t)\over (\partial_s H_{i}(X_{i+}(t), s_{i}))^2}. 
\end{split}
\ee
Here a dot denotes a derivative with respect to $t$.  Integrating, and rewriting in terms of the parameterization from the lower layer, this becomes 
\begin{align*} \tau_{i-}(t)  = \tau_{i-}(0) +  \int_{X_{i+}(0)}^{X_{i+}(t)} \frac{c}{(\partial_s H_i)(\tau, s_i)^2} \, d\tau =  \int_{X_{i+}(0)}^{X_{i+}(t)} \frac{c}{(\dot{X}_{i-}(t'))^2} \, dt', \end{align*}
where we have used the fact that $\tau_{i-}(0) = 0$ by \eqref{alt initial cond}.

Likewise, an ODE characterizing $X_{i+}$ can be derived by differentiating the second identity in \eqref{reparameterization identities}, giving
\begin{align*}
1 & = \Phi'_{i+}(X_{i+}) \dot{X}_{i+} = (1,\eta'_i) \cdot \nabla \varphi_{i+1}(X_{i+}, \eta_i) \dot{X}_{i+}\\
& = \dot{X}_{i+} \sqrt{1+ (\partial_q h_{i})^2} \left. \LB Q_i - 2g (\varrho_{i+1} - \varrho_i) h_{i} + {1+ (\partial_q h_{i})^2 \over (\partial_p h_{i})^2} \RB^{1/2}\right|_{(X_{i+}, p_{i})}.
\end{align*}
Again, initial conditions for this ODE are fixed by \eqref{alt initial cond}, and thus $X_{i+}$ can be defined as the unique solution of 
$$
\left\{\begin{array}{l}
\displaystyle \dot{X}_{i+}(t) = {\dot{X}_{i-} \over \sqrt{\dot{X}_{i-}^2+\dot{Y}_{i}^2}} \left. \LB Q_i - 2g (\varrho_{i+1} - \varrho_i)( Y_{i} + d ) + {c^2 \over \varrho_{i}( Y_{i}^2 + (\dot{X}_{i-})^2)} \RB^{-1/2}\right|_{(X_{i+}, s_{i})},\\\\
X_{i+}(0) = X_{i-}(0).
\end{array}
\right.
$$
In this way,  we can determine $\tau_{i-}$ and $X_{i+}$ from $(X_{i-}, Y_i, \dot{X}_{i-})$.

\subsection{Iteration}

In the previous subsections, we saw how to recover the first internal interface, and how to handle the jump conditions there.  The next step is to iterate this process so that we can reconstruct the entire wave. This is done in two stages:  first, we obtain formulas that hold in the distributional sense.  Second, we use the fact that we are working with $H \in \mathscr{H}$ to deduce that they hold in the strong sense as well.  

%
%

\subsubsection{Iteration in the sense of distributions}
For each $i \geq 1$, we define the trace operator $\Gamma_i$ with domain $\mathscr{H}$ by 
$$
 \Gamma_i: H \mapsto \left( H_i (\cdot, s_i), \, (\partial_s H_i)(\cdot, s_i), \, (\partial_t H_i)(\cdot, s_i) \right) \in \mathscr{H}_i, $$
 where
 $$ \Gamma_i (\mathscr{H}) =: \mathscr{H}_i = \mathscr{H}^{(1)}_i \times \mathscr{H}^{(2)}_i \times \mathscr{H}^{(3)}_i.$$
Simply put, $\Gamma_i$ evaluates the traces of the solution $H$ and its first derivatives in the region $\mathpzc{T}_i$ on the interface $\{ s = s_i \}$. The data we are given to begin the reconstruction process is precisely 
\begin{equation*}\label{H_0 info}
\mathscr{H}_0 := \LC 0, {c \over \sqrt{\varrho_1} \sqrt{c^2 - 2 \mathfrak{p}(t)}}, 0 \RC.
\end{equation*}

The iteration procedure is as follows. In the $(i+1)$-st fluid domain $\mathpzc{T}_{i+1}$, we are solving a Laplace equation with Dirichlet-Neumann data on the bottom boundary $\{s = s_{i}\}$ obtained from previous step: Take the data $( H_{i} (\cdot, s_{i}), \, (\partial_s H_{i})(\cdot, s_{i}), \, (\partial_t H_{i})(\cdot, s_{i}))$ from the solution $H_{i}$ in $\mathpzc{T}_{i}$. In order to recast the hodograph transformation in $\mathpzc{T}_{i+1}$ we need to access to $\tau_{i-}$ computed in the previous subsection. From \eqref{translation}, the continuity of $H$ and $\p_t H$ across the interface $\{ s = s_{i} \}$ can be expressed as
$$
H_{i+1}(t, s_{i}) = H_{i}(\tau_{i-}(t), s_{i}), \quad \p_t H_{i+1}(t, s_{i}) = \p_t H_{i}(\tau_{i-}(t), s_{i}).
$$
This allows us to compute $\p_s H_{i+1}(\cdot, s_{i})$ from \eqref{Bernoulli jump conformal}:
\begin{equation}\label{jump H_s}
\begin{split}
\p_s H_{i+1}(\cdot,s_{i})^2 & = \LB {1\over \p_t H_{i}(\tau_{i-}(\cdot), s_{i})^2 + \p_s H_{i}(\tau_{i-}(\cdot), s_{i})^2} \right. \\
& \qquad \left. + {c^2 - 2g\LB H_{i}(\tau_{i-}(\cdot), s_{i}) -\mathring{H}(s_{i})\RB \over c^2} \LC \varrho_{i+1} - {\varrho}_i \RC \RB^{-1} - \p_t H_{i}(\cdot, s_{i})^2.
\end{split}
\end{equation}
We have therefore found a complete set of Cauchy data $(H_{i+1}, \p_s H_{i+1}, \p_t{H_{i+1}})$ for the Laplace equation \eqref{DN problem} in the next layer $\mathpzc{T}_{i+1}$.  Solving this, gives $( H_{i+1} (\cdot, s_{i+1}), \, (\partial_s H_{i+1})(\cdot, s_{i+1}), \, (\partial_t H_{i+1})(\cdot, s_{i+1}))$.  In total, this procedure can be written abstractly as a mapping 
$$ \mathcal{I}_{i+1}  :=  \mathscr{H}_{i} \to \mathscr{H}_{i+1}, \quad i = 1, \ldots, N-1.$$
The nonlocal and nonlinear expression for $\mathcal{I}_{i+1}$ is given by:
\be \label{formula for I_i} \begin{split}
\mathcal{I}_{i+1}^{(1)}(u,v,w) &= \mathcal{F}^{-1} \LCB\cosh\LB (s_{i+1} - s_{i})\xi \RB \widehat{u \circ \tau_{i-}}(\xi) + {\sinh\LB (s_{i+1} - s_{i})\xi \RB \over \xi} \widehat{f_{i+1}}(u,v,w)(\xi)\RCB  \\
\mathcal{I}_{i+1}^{(2)}(u,v,w) &= \mathcal{F}^{-1} \LCB \xi \sinh\LB (s_{i+1} - s_{i})\xi \RB \widehat{u\circ \tau_{i-}}(\xi) + \cosh\LB (s_{i+1} - s_{i})\xi \RB\widehat{f_{i+1}}(u,v,w)(\xi)\RCB  \\
\mathcal{I}_{i+1}^{(3)}(u,v,w) &= \partial_t (u\circ \tau_{i-}),
\end{split} \ee
where 
\be \label{def f_i} 
f_{i+1}(u,v,w)(\cdot) := \left[ \dfrac{1}{ \dfrac{1}{ v(\tau_{i-}(\cdot))^2 + w(\tau_{i-}(\cdot))^2} - \dfrac{c^2 + 2g\LB u(\tau_{i-}(\cdot))-\mathring{H}(s_{i})\RB}{c^2} (\varrho_{i+1} - \varrho_{i}) } - {w^2}  \right]^{1/2}, \ee
and the Fourier transforms are taken in the sense of distributions. 

\subsubsection{Iteration in the sense of functions}
The above scheme can be improved so that the equalities hold in the sense of functions.  To do this,  we utilize the exponential asymptotics \eqref{asymptotics} to conclude that
\begin{equation}\label{asymptotics i+1}
\begin{array}{l}
H_{i}(\tau_{i-}(t), s_{i}) = H_{i+1}(t, s_{i}) \to \mathring{H}(s_{i}), \quad H_{i+1}(t, s_{i+1}) \to \mathring{H}(s_{i+1}), \\\\
\displaystyle  \p_s H_{i+1}(t, s) \to {1\over \sqrt{\varrho_{i+1}}}, \quad \text{ for } s\in [s_{i}, s_{i+1}], 
\end{array}\quad \textrm{ as }\ |t| \to \infty,
\end{equation}
with the convergence being exponentially fast.  Thus, the differences between the functions and their corresponding asymptotics are square integrable. On the other hand, plugging \eqref{asymptotics i+1} into the first equation of \eqref{formula for I_i}, we find that
\begin{equation*}
\begin{split}
\mathcal{F} \LCB H_{i+1}(\cdot, s_{i+1}) - \mathring{H}(s_{i}) - { s_{i+1} - s_{i} \over \sqrt{\varrho_{i+1}}} \RCB(\xi)
& = \cosh\LB (s_{i+1} - s_{i})\xi \RB \mathcal{F} \LCB H_{i}(\tau_{i-}(\cdot), s_{i}) - \mathring{H}(s_{i}) \RCB(\xi) \\
& \quad + { \sinh\LB (s_{i+1} - s_{i})\xi \RB \over \xi } \mathcal{F} \LCB \p_s H_{i+1} (\cdot, s_{i}) - {1\over \sqrt{\varrho_{i+1}}} \RCB(\xi),
\end{split}
\end{equation*}
where $\p_s H_{i+1} (\cdot, s_{i})$ is determined from \eqref{jump H_s}. Further simplifying the left-hand side  above by using the definition of $\mathring{H}$ in \eqref{limiting H values}, we see that
\begin{equation}\label{recovery H_i+1}
\begin{split}
\mathcal{F} \LCB H_{i+1}(t, s_{i+1}) - \mathring{H}(s_{i+1})  \RCB(\xi) = & \  \cosh\LB (s_{i+1} - s_{i})\xi \RB \mathcal{F} \LCB H_{i}(\tau_{i-}(\cdot), s_{i}) - \mathring{H}(s_{i}) \RCB(\xi) \\
& + { \sinh\LB (s_{i+1} - s_{i})\xi \RB \over \xi } \mathcal{F} \LCB \p_s H_{i+1} (\cdot, s_{i}) - {1\over \sqrt{\varrho_{i+1}}} \RCB(\xi),
\end{split}
\end{equation}
 in the strong sense.

Now, to transform back to the original variables, consider $\mathscr{K}_{i+1}^{-1}$:
$$
x = x_{i+1}(t, s), \quad y = H_{i+1}(t,s) - d.
$$
From the relation $\p_t x_{i+1} = \p_s H_{i+1}$ and \eqref{solution to DN}, we know that
\begin{equation*}
\begin{split}
\mathcal{F}\LCB \p_t x_{i+1} (\cdot, s_{i+1}) \RCB(\xi) = &\ \xi \sinh\LB (s_{i+1} - s_{i})\xi \RB \mathcal{F} \LCB H_{i}(\tau_{i-}(\cdot), s_{i}) \RCB(\xi) \\
&+ \cosh\LB (s_{i+1} - s_{i})\xi \RB \mathcal{F} \LCB \p_s H_{i+1}(\cdot, s_{i}) \RCB(\xi).
\end{split}
\end{equation*}
Similarly, using the asymptotics \eqref{asymptotics i+1}, 
$$
\p_t x_{i+1}(t, s_{i+1}) = \p_s H_{i+1} (t, s_{i+1}) \to {1\over \sqrt{\varrho_{i+1}}} \quad \textrm{ exponentially as } \ |t| \to \infty,
$$
and the identity
$
{\xi\sinh (a \xi) }\mathcal{F}\{1\}(\xi) = {\xi \sinh (a \xi)} 2\pi \delta(\xi) = 0,\ 
$
we ultimately conclude that 
\begin{equation}\label{recovery theta_i+1}
\begin{split}
\mathcal{F} \LCB \p_t x_{i+1}(\cdot, s_{i+1}) - {1\over \sqrt{\varrho_{i+1}}}  \RCB(\xi) = & \  \xi \sinh\LB (s_{i+1} - s_{i})\xi \RB \mathcal{F} \LCB H_{i}(\tau_{i-}(\cdot), s_{i}) - \mathring{H}(s_{i}) \RCB(\xi) \\
& +\cosh\LB (s_{i+1} - s_{i})\xi \RB \mathcal{F} \LCB \p_s H_{i+1}(\cdot, s_{i}) - {1\over \sqrt{\varrho_{i+1}}} \RCB(\xi).
\end{split}
\end{equation}
Note that $X_{(i+1)-}(t) = x_{i+1}(t, s_{i+1})$ and $Y_{i+1}(t) = H_{i+1}(t, s_{i+1})$. Hence the $(i+1)$-st interface can be parametrized by
\begin{equation}\label{x i+1 layer}
\begin{split}
X_{(i+1)-}(t) = {t\over \sqrt{\varrho_{i+1}}} + \int^t_{-\infty} \mathcal{F}^{-1} & \LCB \xi \sinh\LB (s_{i+1} - s_{i})\xi \RB \mathcal{F} \LB H_{i}(\tau_{i-}(t), s_{i}) - \mathring{H}(s_{i+1}) \RB(\xi) \right. \\
& \ \left. +\cosh\LB (s_{i+1} - s_{i})\xi \RB \mathcal{F} \LB \p_s H_{i+1}(t, s_{i}) - {1\over \sqrt{\varrho_{i+1}}} \RB(\xi) \RCB(t')\ dt'
\end{split}
\end{equation}
\begin{equation}\label{eta i+1 layer}
\begin{split}
Y_{i+1}(t) = \mathcal{F}^{-1} & \LCB \cosh\LB (s_{i+1} - s_{i})\xi \RB \mathcal{F} \LB H_{i}(\tau_{i-}(t), s_{i}) - \mathring{H}(s_{i+1}) \RB(\xi) \right. \\
& \ \left. + { \sinh\LB (s_{i+1} - s_{i})\xi \RB \over \xi } \mathcal{F} \LB \p_s H_{i+1} (t, s_{i}) - {1\over \sqrt{\varrho_{i+1}}} \RB(\xi) \RCB(t) + \mathring{H}(s_{i+1}) - d.
\end{split}
\end{equation}
An argument like that given at the end of Section \ref{subsec_firstlayer} shows that these are equalities in the sense of functions.

In summary, to iterate in the regular function setting, we replace the trace operator $\Gamma_{i}$ in the previous subsection by another one which evaluates the difference between the trace and its asymptotics
\begin{equation*}
\begin{split}
\tilde{\Gamma}_{i}(\mathscr{H}) & = \left( H_{i} (\cdot, s_{i}) - \mathring{H}(s_{i}), \, (\partial_s H_{i})(\cdot, s_{i}) - {1\over\sqrt{\varrho_{i}}}, \, (\partial_t H_{i})(\cdot, s_{i}) \right) =: \tilde{\mathscr{H}}_{i}, \ i = 1,\ldots, N.
\end{split}
\end{equation*} 
Likewise, we normalized the data on the bed so that it has a Fourier transform defined in the strong sense:
\begin{equation*}
\tilde{\mathscr{H}}_0 := \LC 0, {c \over \sqrt{\varrho_1} \sqrt{c^2 - 2 \mathfrak{p}(t)}} - {1\over\sqrt{\varrho_1}}, 0 \RC.
\end{equation*}
A corresponding reconstruction map for the height function $\tilde{\mathcal{I}}_{i+1} : \tilde{\mathscr{H}}_{i} \to \tilde{\mathscr{H}}_{i+1}, \ i= 1, \ldots, N-1$ can thus be constructed as follows.
\begin{equation}\label{iteration function}
\begin{split}
& \tilde{\mathcal{I}}^{(1)}_{i+1} = \mathcal{F}^{-1}\LCB \textrm{the right-hand side of \eqref{recovery H_i+1}} \RCB, \ \textrm{where $\p_s H_{i+1}(\cdot,s_i)$ is given in \eqref{jump H_s}},  \\
& \tilde{\mathcal{I}}^{(2)}_{i+1} = \p_t x_{i+1}(t, s_{i+1}) - {1\over \sqrt{\varrho_{i+1}}} = \mathcal{F}^{-1}\LCB \textrm{the right-hand side of \eqref{recovery theta_i+1}} \RCB,  \\
& \tilde{\mathcal{I}}^{(3)}_{i+1} = \p_t \tilde{\mathcal{I}}^{(1)}_{i+1},
\end{split}
\end{equation}
which  the corresponding iteration scheme for the parametrized interface reconstruction is 
\begin{equation}\label{iteration surface}
\left\{\begin{array}{l}
\displaystyle X_{(i+1)-}(t) =  {t\over \sqrt{\varrho_{i+1}}} + \int^t_{-\infty} \tilde{\mathcal{I}}^{(2)}_{i+1}(t')\ dt',\\\\
Y_{i+1}(t) = \tilde{\mathcal{I}}^{(1)}_{i+1}(t) + \mathring{H}(s_{i+1}) - d,
\end{array}\right. \quad i = 1, \ldots, N-1.
\end{equation}
Composition then gives the $(i+1)$-st layer directly from the data on the ocean floor (which can be determined from the bottom pressure readings):
$$
\mathcal{J}_{i+1} := \tilde{\mathcal{I}}_{i+1}\circ \tilde{\mathcal{I}}_{i} \circ \cdots \circ \tilde{\mathcal{I}}_{1} : \tilde{\mathscr{H}}_0 \to \tilde{\mathscr{H}}_{i+1}.
$$
and the reconstruction of the air-sea interface can thus be obtained at $i = N-1$. 
\section{Reconstruction for continuous stratification}\label{continuous section} 

\subsection{Overview}
We now consider the task of reconstructing the free surface of a \emph{continuously} stratified traveling wave from pressure data.  The method of the previous section was to adopt conformal coordinates where the governing equation is simply Laplace in a strip.  As continuously stratified waves are rotational, the hodograph transform is no longer available, and the dynamics in the interior are too complicated to permit a similar approach.  With that in mind, we will instead proceed by \emph{approximating} the continuous wave by a layer-wise constant density and irrotational one.  

Before giving the details of this procedure, let us recall some terminology and important results from the literature.  The existence of small-amplitude supercritical solitary waves in both the continuous stratification and layer-wise smooth stratification regimes can be found in Turner \cite{turner1981internal,turner1984variational}, where they are built via a limit of periodic solutions with periods increasing to $\infty$.  Supercritical here refers to the fact that the wave speed is sufficiently large so that the (generalized) Froude number for the flow lies below $1$. 
Turner (and all of the contemporaneous existence theory) dealt exclusively with steady waves for which the Bernoulli function $\beta$ has the form \eqref{solitary beta}.  As we have discussed previously, this is the only possibility if one considers solitary waves that have uniform velocity upstream and downstream.  Because he aimed to construct such waves as limits of periodic solutions, Turner naturally looked for periodic waves that also have this type of Bernoulli function.   It will be important for our later analysis to note that both the periodic and solitary waves in Turner's theory are waves of elevation, that is, their streamlines lie above those for the corresponding laminar flow.  More precisely, if we define $\mathring{h}$ as in \eqref{limiting h}, then a wave of elevation is one whose height function $h$ satisfies
\be h - \mathring{h} > 0 \qquad \textrm{in } \overline{\mathcal{R}} \setminus \{ p = p_0 \}. \label{h wave of elevation} \ee





Our reconstruction method is based upon the fact that small-amplitude periodic waves with $\beta$ as in \eqref{solitary beta} depend continuously on the streamline density function.  This was proved recently by the authors in \cite{ChenWalsh2014continuity}.  Briefly, the result is the following. 
\begin{thm}[Chen--Walsh \cite{ChenWalsh2014continuity}] \label{continuity theorem} Fix a H\"older exponent $\alpha \in (0,1)$, a period $L$ sufficiently large, a depth $d>0$, and a pseudo volumetric mass flux $p_0<0$. Let ${\rho}_* \in C^{1,\alpha}([p_0, 0])$ with $\rho_*(0) = 1$ be a stable streamline density function and suppose that $h_*$ is a corresponding solution of \eqref{h eqn} for a wave speed $c_*$ that is supercritical.  If $h_*$ is a wave of elevation and sufficiently small-amplitude, then the following statements hold.
\begin{enumerate}
\item[(A1)] There is a neighborhood ${\mathcal{U}}$ of ${\rho}_*$ in $L^\infty([p_0,0])$ such that for any ${\rho} \in {\mathcal{U}}$ with $\rho(0) = 1$ that is non-increasing and piecewise smooth, there is a solution $(h, {\rho}, c)$ to \eqref{h eqn} such that $h$ is a strict wave of elevation and $c$ is supercritical.  
\item[(A2)] We have the estimates
\begin{equation}\label{conv_h}
\|h - h_*\|_{W_{\textrm{per}}^{1, r}(\overline{\mathcal{R}})} \leq C_1 \|{\rho} - {\rho}_*\|_{L^\infty}
\end{equation}
and, 
\begin{equation}\label{conv_c}
c - c_* = {1\over d} \int^0_{p_0} \LB {1\over \sqrt{\rho}} - {1\over \sqrt{\rho_*}} \RB\ ds = \mathcal{O}(\|\rho - \rho_*\|_{L^\infty}).
\end{equation}
for some constant $C_1$ depending on $\rho_*$, $h_*$, and $L$.  Here $r := 2/(1-\alpha)$ is chosen so that $W^{1,r}$ embeds continuously in $C^{0,\alpha}$.  
\item[(A4)] Moreover, the pressure trace on the bed converges in the following sense.  Let a connected set $I \subset\subset [p_0, 0]\setminus \{ p_1, \ldots, p_{N-1}\}$ be given with $p_0 \in I$, and assume that $\rho \in C^{1,\alpha}(\overline{I})$. Denote by  $P_{\textrm{bot}}$  the trace of the pressure on the ocean bed for the traveling wave with density $\rho$, and let $P_{{\textrm{bot}}*}$ be the trace of $P_*$ on the bed.    Then
\begin{equation}\label{conv_P} 
\| P_{\textrm{bot}} - P_{\textrm{bot}*} \|_{C_{\textrm{per}}^{0, \alpha}(\mathbb{R})} \leq C_2 \left(   \| \rho - \rho_*\|_{L^\infty([p_0,0])} +  \| \rho - \rho_*\|_{C^{1,\alpha}(\overline{I})} \right),
\end{equation}
where $C_2 > 0$ depends on the length of $I$, $\rho_*$, $h_*$, and $L$.  
\end{enumerate}
\end{thm}

The basic structure of the approximation scheme is intuitive, but several nontrivial technical complications arise.  Take a solitary wave with continuous stratification that is small-ampltiude, supercritical, and a wave of elevation.  Assume, as in Turner, that this wave can be realized as the limit of periodic solutions of increasing period.  Each of these periodic waves can in turn be approximated by a layer-wise constant density irrotational flow using Theorem \ref{continuity theorem}.  Finally, the multi-layered flows can be reconstructed from the pressure data as detailed in Section \ref{streamline recovery section}.  

\subsection{Layer-wise approximation to continuous stratification}\label{subsec layer approx}
In order to rigorously implement the approach suggested above, however, we need to more closely examine the convergence of periodic waves to solitary waves.  Unfortunately, this also requires working in a new coordinate system that can be seen as a normalized version of the semi-Lagrangian variables.  Specifically, we wish to non-dimensionalize and rescale the streamline coordinate $p$ by mapping
$$ (q, p) \mapsto  (\xi, \zeta),\quad 
\text{where } \quad
\xi := \frac{q}{d}, \quad \zeta := \frac{1}{cd} \int_0^p \frac{1}{\sqrt{\rho(s)}} \, ds.$$
For each $L > 0$, the rectangle $\mathcal{R}^L := \{ (q,p) \in (-L, L) \times (p_0, 0) \}$ is then taken to
$$ \mathcal{S}^L := \{ (\xi, \zeta) : \xi \in (-\frac{L}{d}, \frac{L}{d}), ~ \zeta \in (-1, 0) \},$$
with each layer $\mathcal{R}_i^L$ being sent to
$$ \mathcal{S}_i^L := \{ (\xi, \zeta) : \xi \in (-\frac{L}{d}, \frac{L}{d}), ~ \zeta \in (\zeta_{i-1}, \zeta_i) \}, \qquad \textrm{where } \zeta_i := \frac{1}{cd} \int_0^{p_i} \frac{1}{\sqrt{\rho(s)}} \, ds.$$
We will continue to denote the full strip by $\mathcal{R}$.   Likewise, $\mathcal{R}_i$, $\mathcal{S}$, and $\mathcal{S}_i$ are each unbounded.  

The height equation \eqref{h eqn} is then recast in terms of a new unknown
\be w = w(\xi, \zeta) := \dfrac{y(\xi, p(\zeta))}{d} - \zeta, \label{def w} \ee
which will satisfy another quasilinear elliptic problem with the rescaled streamline density function
\be \mathring{\rho}(\zeta) := \rho(p(\zeta)) \label{def scaled rho} \ee
appearing as a coefficient.  Here, $w$ measures the deviation of the streamlines from their height at infinity.  One advantage of using this quantity is that it decays upstream and downstream and thus lies in a Sobolev space set on $\mathcal{S}$.  Note that $w$ is a wave of elevation \eqref{h wave of elevation} provided that 
\be w > 0 \qquad \textrm{in } \overline{\mathcal{S}} \setminus \{ \zeta = -1 \}.\label{w wave of elevation} \ee  

Following Turner, to $w$ as in \eqref{def w}, we associate an energy over the set $\mathcal{K} \subset \mathcal{S}$ 
\begin{equation}\label{energy} 
F(w, \mathcal{K}) := \int_{\mathcal{K}} \mathring{\rho} {|\nabla w|^2 \over 1 + w_\zeta}\ d\xi \, d\zeta.
\end{equation}
Note that the absence of stagnation \eqref{nostagnation} requirement translates to the statement that $1+w_\zeta > 0$ in these variables, hence the above integral is well-defined.   Physically, $F$ describes the amount of kinetic energy imparted by the deviation from the far field state.  

Let us now more precisely formulate the class of solitary waves for which we can successfully prosecute the scheme outlined above.

\begin{assumption} \label{recoverable solitary wave}   Consider a solitary stratified water wave represented by $(w_s, \mathring{\rho}, c_s)$.  We require that (i) the wave has finite energy
$$ F(w_s, \mathcal{S}) = R^2 < \infty,$$
that is sufficiently small; (ii) $\| w_s \|_{W^{1,\infty}(\mathcal{S})} < \delta$ for some $\delta > 0$ sufficiently small;  (iii) $w_s$ is a wave of elevation \eqref{w wave of elevation}; and (iv) $w_s$ can be realized as a limit of periodic stratified waves with the same energy.  This last statement means that there exists a sequence $\{ (w_m, \mathring{\rho}, c_m) \}$ of continuously stratified periodic steady waves with $w_m, \partial_\xi w_m$ uniformly exponentially localized, 
$$ \| \nabla w_m \|_{W^{1,\infty}(\mathcal{S}^m)} < \delta, \qquad F(w_m, \mathcal{S}^m) = R^2,$$
such that there is a subsequence converging to $w_s$ in $C^{0,\alpha}(\mathcal{S}^k)$ for all $k \in \mathbb{N}$ and $c_m \to c$.  
\end{assumption}
From Turner \cite{turner1984variational}, we know that waves satisfying Assumption \ref{recoverable solitary wave} exist.  Indeed, much of the solitary wave existence theory for stratified waves is built on exactly this type of limiting construction.  Nonetheless, one cannot say that {every} solitary stratified wave falls within this class, hence we must make it a hypothesis.   

Now let $(w_s, \mathring{\rho}, c_s)$ be given as above with $\{ (w_m, \mathring{\rho}, c_m)\}$  the corresponding sequence of continuously stratified periodic waves.  Then we may apply Theorem \ref{continuity theorem} to approximate each $(w_m, \mathring{\rho}, c_m)$ by a layer-wise constant density irrotational wave $(w^N_m, \mathring{\rho}^N, c^N_m)$.  From \eqref{conv_h}--\eqref{conv_c} we have the estimates  
$$
\|w^N_m - w_m\|_{C^{0,\alpha}_{\textrm{per}}(\overline{\mathcal{R}^m})} \leq C\|w^N_m - w_m\|_{W^{1,r}_{\textrm{per}}({\mathcal{R}^m})} \leq C_1 \|\mathring{\rho}^N - \mathring{\rho}\|_{L^\infty},$$
and
$$ | c^N_m - c_m |  \leq C \|\mathring{\rho}^N - \mathring{\rho}\|_{L^\infty},$$
where $r = 2/(1-\alpha)$ and $C_1 = C_1(w_m, \mathring{\rho}, m)$. Note that using \eqref{conv_c}, we can adjust $\mathring{\rho}^N$ to ensure that $c^N_m = c_m$, and thus the multi-layer  solutions are $(w^N_m, \mathring{\rho}^N, c_m)$.

The difference in the energies can then be estimated:
\begin{equation*}
\begin{split}
\LV F(w_m^N, \mathcal{S}^m) - F(w_m, \mathcal{S}^m) \RV & \leq  \int_{\mathcal{S}^m} \LV \mathring{\rho}^N {|\nabla w^N_m|^2 \over 1 + \partial_\zeta w^N_m} - \mathring{\rho}_s {|\nabla w_m|^2 \over 1 + \partial_\zeta w_m}\RV d\xi \, d\zeta\\
&\leq C_2 \|\mathring{\rho}^N - \mathring{\rho}_s\|_{L^\infty} + C_3 \|\mathring{\rho}_s\|_{L^\infty} \|w^N_m - w_m\|_{W^{1,2}(\mathcal{S}^m)}\\
&\leq \LC C_2 + C_1C_3 \|\mathring{\rho}_s\|_{L^\infty} \RC \|\mathring{\rho}^N - \mathring{\rho}_s\|_{L^\infty},
\end{split}
\end{equation*} 
where $C_2 = C_2(\|\nabla w^N_m \|_{L^\infty(\mathcal{S}^m)}, \|\nabla w^N_m\|_{L^2(\mathcal{S}^m)})$ and $C_3 = C_3(\|\nabla w^N_m \|_{L^\infty(\mathcal{S}^m)}, \|\nabla w_m \|_{L^\infty(\mathcal{S}^m)})$. From \cite[Lemma 3.7, Lemma 3.10]{ChenWalsh2014continuity} we know that,
$$
\|\nabla w^N_m \|_{L^\infty(\mathcal{S}^m)}, \|\nabla w^N_m\|_{L^2(\mathcal{S}^m)} \leq C(\|\nabla w_m\|_{L^2(\mathcal{S}^m)}) \leq C(R, \delta)
$$
uniformly in $N$.  The exponential localization implies moreover that these constants are uniform in $m$ as well. Therefore
\begin{equation}\label{conv_F}
\LV F(w_m^N, \mathcal{S}^m) - R^2 \RV \leq C \|\mathring{\rho}^N - \mathring{\rho}\|_{L^\infty}.
\end{equation}

Now for any fixed $N$, following the diagonalization technique as in Turner \cite{turner1981internal}, modulo a subsequence, we have
$$
w^N_m \to w^N_s, \ \textrm{ as } m\to \infty, \qquad \textrm{in } C^{0, \alpha}(\mathcal{R}^k) \textrm{ for all } k\in \mathbb{N},
$$
where $(w^N_s,\mathring{\rho}, c^N_s)$ is a layer-wise smooth solitary wave carrying an energy 
$$
F(w^N_s, \mathcal{S}) = \lim_{m\to\infty} F(w^N_m, \mathcal{S}^m).
$$
The convergence of the right-hand side is a consequence of the uniform (in $m$) exponential decay of $w^N_m$ and $\nabla w^N_m$. From \eqref{conv_F} we know that 
$$
\LV F(w_s^N, \mathcal{S}) - F(w_s, \mathcal{S}) \RV \leq C \|\mathring{\rho}^N - \mathring{\rho}\|_{L^\infty}.
$$
Passing to a further subsequence ensures the convergence of $c^N_m = c_m$ to some $c^N_s$. So, for any $\epsilon > 0$, taking $m$ large enough gives
\begin{equation*}
\begin{split}
\|w^N_s - w_s\|_{C^{0,\alpha}(\overline{\mathcal{R}_k})} & \leq \|w^N_m - w^N_s\|_{C^{0,\alpha}(\overline{\mathcal{R}^k})} + \|w^N_m - w_m\|_{C^{0,\alpha}(\overline{\mathcal{R}^k})} + \|w_m - w_s\|_{C^{0,\alpha}(\overline{\mathcal{R}_k})}\\
& < 2\varepsilon + C \|\mathring{\rho}^N - \mathring{\rho}_s\|_{L^\infty},
\end{split}
\end{equation*}
for all $k\in \mathbb{N}$. Finally, letting $m\to \infty$ we conclude
\begin{equation}\label{solitary convergence}
\|w^N_s - w_s\|_{C^{0,\alpha}(\overline{\mathcal{R}_k})} \leq C \|\mathring{\rho}^N - \mathring{\rho}\|_{L^\infty}, \quad \textrm{ for all } k \in \mathbb{N}. 
\end{equation}

Therefore, given the pressure data for $w_s^N$, we can apply the reconstruction procedure of the previous subsection to recover $w_s$ up to $\mathcal{O}(\|\mathring{\rho}^N - \mathring{\rho}\|_{L^\infty})$ in $C^{0,\alpha}$.  


\subsection{Convergence of the ocean bed parameterization}  
The argument above shows that we can approximate the continuously stratified solitary wave via a sequence of layer-wise constant density irrotational solitary waves.  However, its main weakness is that it requires the pressure data not for the continuously stratified solitary wave (which is observed), but for a limiting (sub)sequence which are only known to exist mathematically.  This is partially ameliorated by \eqref{conv_P} which guarantees that the pressure traces of the subsequence converge to that of the continuous solitary wave.  

But there are two issues remaining that must be addressed.  The first of these is rather benign:  in Section \ref{streamline recovery section}, we begin the reconstruction process with the data $\mathfrak{p}$.  It is therefore necessary to confirm that the convergence of the pressure likewise extends to this quantity.     Note that because $\mathfrak{p}$ is defined in conformal variables --- which do not have an obvious meaning for continuously stratified flows --- it is not immediately clear that this should be true.   The second more serious issue is the sensitivity of the reconstruction scheme to noise in the pressure data.  This we discuss in the next subsection.

Let $h$ be the height function in semi-Lagrangian coordinates for a steady wave localized near the crest (it can be solitary or periodic).  Suppose that, as above, there is a sequence $\{ h^N \}$ of layer-wise constant density steady wave localized near the crest that converges to $h$.  Then for each $h^N$, the conformal coordinates are well-defined and we have 
$$ t = - \frac{1}{c} \varphi^N(x_0^N(t), -d),$$
and hence
$$ -c = \varphi_x^N(x_0^N(t), -d) \dot{x}_0^N(t) = - \psi_y(x_0^N(t), -d) \dot{x}_0^N(t).$$
But then, using the change of variables formulas, this implies that $x_0^N$ solves the ODE 
$$ \dot{x}_0^N = h_p^N(x_0^N, p_0), \qquad x_0^N(0) = 0.$$

In view of this, we define $x_0 = x_0(t)$ to be the solution of
\be \dot{x}_0(t) = h_p(x_0(t), p_0), \qquad x(0) = 0, \label{ODE for x_0} \ee
where the dot denotes differentiation with respect to $t$.     Formally, the initial value problem \eqref{ODE for x_0} is the limit of the ODE satisfied by $x_0^N$ as $N \to \infty$.  We hope to show that $x_0^N$ likewise converges to $x_0$, which then implies $\mathfrak{p}^N$ approaches $\mathfrak{p}$.  

\begin{prop}  Let $\{\rho^N\} \subset L^\infty([p_0,0])$ be a sequence of piecewise constant stable streamline density functions, with $\rho^N \to \rho$ in $L^\infty$, where $\rho \in C^{1,\alpha}([p_0,0])$ is a stable smooth streamline density function that is constant on some neighborhood $I$ of the bed.  Let  $h$ and $h^N$ be the height functions for a traveling wave localized near the crest with period $L \in (0, \infty]$, the speed $c$ and $c^N$, and streamline density function $\rho$ and $\rho^N$ respectively. Assume that  
$$ c\sqrt{\rho(p_0)} = c^N \sqrt{\rho^N(p_0)}, $$
and that there exist constants $M, \gamma > 0$ independent of $L$ and $N$ such that
$$ | h_p(q, p_0) - \mathring{h}_p(p_0) |, ~ | h_p^N(q, p_0) - \mathring{h}_p(p_0) | \leq M  e^{-\gamma |q|}.$$
If $L = \infty$, i.e., the waves are solitary, also suppose that $h_p^N(\cdot, p_0) \to h_p(\cdot, p_0)$ in $L^\infty(\mathbb{R})$.  Then  
it follows that 
$$ x_0^N \to x_0 \qquad \textrm{in } L^\infty(\mathbb{R}).  $$
\end{prop} 
\begin{proof}
Denote $z^N := x_0 - x_0^N$, so that $z^N$ satisfies 
\be z^N(t) = \int_0^t [ h_p(x_0(\tau), p_0) - h_p^N(x_0^N(\tau), p_0) ] \, d\tau. \label{z^N integral equation} \ee
The no stagnation condition implies that there exists $\sigma > 0$ such that 
$$ h_p(\cdot, p_0), \, h_p^N(\cdot, p_0) \geq \sigma > 0, \textrm{ for all $N \geq 1$.}  $$
Recalling \eqref{ODE for x_0}, this guarantees that
$$ x_0(t), \, x_0^N(t) \geq \sigma t, \qquad \textrm{for all } t \in [0,\infty).$$

Fix $\epsilon > 0$, and choose $T > 0$ sufficiently large so that 
$$ \int_T^\infty e^{-\sigma \gamma t} \, dt < \frac{\epsilon}{4M}.$$
Starting from \eqref{z^N integral equation}, we write 
\begin{align*} \| z^N \|_{L^\infty} &\leq \int_0^T  | h_p(x_0(\tau), p_0) - h_p^N(x_0^N(\tau), p_0) | \, d\tau + \int_T^\infty  [ h_p(x_0(\tau), p_0) - h_p^N(x_0^N(\tau), p_0) ] \, d\tau \\
& =: \mathbf{I} + \mathbf{II}.\end{align*}
To estimate $\mathbf{II}$, we exploit the fact that $\mathring{h}_p(p_0) = \mathring{h}_p^N(p_0)$, and the exponential localization: 
\begin{align*}
\mathbf{II} & \leq \int_T^\infty  |h_p(x_0(\tau), p_0) -  \mathring{h}_p(p_0)| \, d\tau + \int_T^\infty  | h_p^N(x_0^N(\tau), p_0) - \mathring{h}_p^N(p_0)|  \, d\tau \\
& \leq 2M \int_T^\infty e^{-\gamma \sigma t} \, dt < \frac{\epsilon}{2}.
\end{align*}

On the other hand, for $\mathbf{I}$ we use the simple estimate 
$$ \mathbf{I} \leq T \| h_p(\cdot, p_0) - h_p^N(\cdot, p_0) \|_{L^\infty} \leq T \| h_p(\cdot, p_0) - h_p^N(\cdot, p_0) \|_{L^\infty} < \frac{\epsilon}{2}, $$ 
for $N$ sufficiently large. Observe that in the periodic case $L < \infty$, we have 
$$\| h_p(\cdot, p_0) - h_p^N(\cdot, p_0) \|_{L^\infty} \leq C \left( \| \rho^N - \rho \|_{L^\infty([p_0,0])} +  \| \rho^N - \rho \|_{C^{1,\alpha}(I)}  \right),$$
and hence the $L^\infty$ convergence of $h_p^N$ on the bed is implied by our previous assumptions.

The argument above gives the convergence of $x_0^N$ to $x_0$ in $L^\infty([0,\infty))$; the argument for $L^\infty((-\infty, 0]))$ is similar.  \end{proof}

\subsection{Pressure noise} \label{pressure noise section}
Let us now discuss the sensitivity of the above considerations to pressure data.  First, it is important to note that any scheme that is attempting to reconstruct from pressure must come to terms with the physical limitations of pressure transducers.  It is of course impossible to perfectly read the pressure, and in particular, existing instruments fail to capture very high frequency components.  Likewise, computers are finite precision machines, and so any practical implementation of the recovery scheme will necessitate additional approximation.  For instance, it is common to use a windowed transfer functions that simply ignores the high frequency range above a certain cut-off. Mathematically, this procedure corresponds to applying a Littlewood-Payley projection onto a finite frequency ball.   However, because we are working with an ill-posed PDE system, it is very challenging to prove that convergence occurs as the cut-off radius is increased.   

This seems endemic to all reconstruction from pressure procedures, even for constant density irrotational flows.  Nonetheless, actual numerical experiments suggests that the schemes converge, and give reasonable agreement with observation \cite{deconinck2012relating,oliveras2012recovering}.  

Unfortunately, the reconstruction process for continuously stratified fluids given above possesses  an additional source of error.   Observe that in our approximation we take the pressure reading of a solitary wave with continuous stratification and treat this as pressure data for a layer-wise constant density irrotational wave.  But,  the unique determinability result indicates that the pressure trace for these two waves cannot possibly coincide.  We can therefore think of this as attempting to run the reconstruction scheme of Section \ref{streamline recovery section} with noisy pressure data.  

To be more concrete, denote
\be\label{truepressure}
P_{\text{bot}}^N = P_{\text{bot}*} + P_{\textrm{err}},
\ee
where $P_{\text{bot}*}$ is the bottom pressure reading of the solitary wave with continuous stratification, and $P_{\text{bot}}^N$ is the trace of the pressure for the $N$-layered solitary wave. From \eqref{bottom pressure} we then have
\be\label{approx pressure}
\mathfrak{p}^N = \mathfrak{p} + {P_{\textrm{err}} \over \varrho_1} + {1\over \varrho_1}\int^0_{p_0} g \LB \rho(p)  - \rho^N(p) \RB\ dp=: \mathfrak{p} + \mathfrak{p}_{\textrm{err}}.
\ee
In order to recover the true $N$-layered flow, one should use $\mathfrak{p}^N(t)$ in the iteration scheme \eqref{iteration surface}, rather than $\mathfrak{p}(t)$. Hence the error from our iteration scheme consists of two parts: one from the limiting process by taking large number of layers; the other from the pressure error. 

Let
$$
f(\mathfrak{p}) = {1\over \sqrt{\varrho_1}} \LC {c\over  \sqrt{c^2 - 2 \mathfrak{p}}} - 1 \RC.
$$
From \eqref{eta 1st layer} the error in the reconstruction of the first interface is
\begin{equation}\label{error in eta}
\begin{split}
Y^N_1- Y_1 & = \mathcal{F}^{-1}\LCB { \sinh \LB (s_1-s_0)\xi \RB \over \xi} \mathcal{F}\LB  f(\mathfrak{p}^N) - f(\mathfrak{p}) \RB(\xi) \RCB.
\end{split}
\end{equation}
The difference term can be computed as
\begin{equation*}
\begin{split}
f(\mathfrak{p}^N) - f(\mathfrak{p}) & = {2c \over \sqrt{ \varrho_1(c^2 - 2\mathfrak{p}) (c^2 - 2\mathfrak{p}^N) } \LC  \sqrt{c^2 - 2 \mathfrak{p}} +  \sqrt{c^2 - 2 \mathfrak{p}^N} \RC} \cdot (\mathfrak{p} - \mathfrak{p}^N) 
\end{split}
\end{equation*}

To justify the inverse Fourier transform, or more precisely, to make sense of $Y^N_1$ and $Y_1$ as functions, one needs to ensure that $\mathcal{F}\LB f(\mathfrak{p}^N)\RB$ and $\mathcal{F}\LB f(\mathfrak{p})\RB$ have sufficient decay, e.g. 
\begin{equation*}
e^{(s_1-s_0)|\xi|}\mathcal{F}\LB f(\mathfrak{p}^N) \RB(\xi), \quad e^{(s_1-s_0)|\xi|}\mathcal{F}\LB f(\mathfrak{p}) \RB(\xi) \in L^2(\R).
\end{equation*}
In fact, this can be proved rigorously if the continuously stratified wave has constant density in $\mathcal{S}_1$.  Following Clamond-Constantin \cite{constantin2013pressure}, we see that the pressure transfer function $\mathfrak{p}^N(t)$ can be holomorphically extended to a function $\mathfrak{p}^N(z)$ with $z = t + is$ defined in the strip $\{ (t, s)\in \R \times (2s_0 - s_1 - \varepsilon, s_1 + \varepsilon) \}$ for some $\varepsilon > 0$.  Likewise, the analyticity of  $\mathfrak{p}^N(z)$ and the no stagnation condition imply that the function 
$$
f(\mathfrak{p}^N(z)) =  {1\over \sqrt{\varrho_1}} \LC {c \over \sqrt{c^2 - 2 \mathfrak{p}^N(z)}} - 1 \RC 
$$
is analytic in $\R\times [-\lambda, \mu]$, where $\lambda = \mu = s_1 - s_0 + {\varepsilon\over 2}$. Thus, the Paley--Weinder Theorem (cf., e.g. \cite[Theorem 4]{PaleyWiener}) implies that 
\begin{equation*}
e^{(s_1-s_0)|\xi|}\mathcal{F}\LB f(\mathfrak{p}^N) \RB(\xi) \in L^2(\R).
\end{equation*}

However the above argument does not extend beyond the region of constant density, and hence the decay of $\mathcal{F}\LB f(\mathfrak{p}) \RB(\xi)$ becomes very hard to verify for more general stratification.  Worse still, because one loses analyticity crossing the internal interfaces, it is unclear how this reasoning could be used in any layer but the first.   For this reason, we cannot currently prove that our method represents a convergent numerical scheme.  On the other hand, as it does succeed an analytical approximation method, we anticipate that it will have applications for studying the qualitative behavior features of continuously stratified solitary waves.

\bibliographystyle{siam}
\bibliography{projectdescription}
\end{document}